\documentclass[11pt,letterpaper]{amsart}
\usepackage[foot]{amsaddr}
\usepackage[T1]{fontenc}
\usepackage{etoolbox}

\makeatletter
\patchcmd{\@sect}{\@addpunct.}
  {\ifstrequal{#1}{paragraph}{}{\@addpunct.}}{}{
  \PackageError{manuscript}{Could not customize paragraph punctuation}
    {Check the definition of \string\@sect\space in the document class.}}
\def\paragraph{\@startsection{paragraph}{4}%
  \z@{.5\linespacing\@plus.7\linespacing}{-\fontdimen2\font}%
  \normalfont}
\makeatother

\usepackage{ifxetex}
\ifxetex
  \usepackage[no-math]{fontspec}
\else
\fi
\usepackage{amsmath}
\usepackage{amsfonts}
\usepackage{amssymb}
\usepackage{amsthm}
\usepackage{fullpage}
\usepackage{tablefootnote}
\usepackage{hyphenat}
\usepackage{microtype}
\RequirePackage{silence}
\usepackage{hyphenat}
\ifxetex
  \usepackage[libertine]{newtxmath}
\else
  \usepackage{newtxmath}
\fi
\usepackage[tt=false]{libertine}
\usepackage{caption}
\usepackage{tcolorbox}
\tcbuselibrary{breakable,skins}
\usepackage{bbm}
\usepackage{hyperref, color}
\hypersetup{colorlinks=true,citecolor=blue, linkcolor=blue, urlcolor=blue}
\usepackage[linesnumbered,boxed,ruled,vlined]{algorithm2e}
\usepackage{bm}
\usepackage{bbm}
\usepackage[numbers]{natbib}
\usepackage{xcolor}
\usepackage{tikz}
\usepackage{enumerate}
\usepackage{enumitem}
\usepackage{ragged2e}
\usepackage{tabularx}
\usepackage{array}
\usepackage{longtable}
\usepackage{hhline}
\usepackage{multirow}
\newcolumntype{L}[1]{>{\=Raggedright\arraybackslash}p{#1}}
\newcolumntype{C}[1]{>{\centering\arraybackslash}m{#1}}
\newcolumntype{R}[1]{>{\=Raggedleft\arraybackslash}p{#1}}

\usepackage{makecell}
\usepackage{aliascnt}
\usepackage{prettyref}
\usepackage{footnote}
\usepackage{float}
\usepackage{fix-cm}
\makesavenoteenv{tabular}

\newtheorem{theorem}{Theorem}[section]
\newcommand{\newsharedtheorem}[2]{%
  \newaliascnt{#1}{theorem}%
  \newtheorem{#1}[#1]{#2}%
  \aliascntresetthe{#1}%
}
\newsharedtheorem{observation}{Observation}
\newsharedtheorem{claim}{Claim}
\newtheorem*{claim*}{Claim}
\newsharedtheorem{condition}{Condition}
\newsharedtheorem{example}{Example}
\newsharedtheorem{fact}{Fact}
\newsharedtheorem{lemma}{Lemma}
\newsharedtheorem{proposition}{Proposition}
\newsharedtheorem{conjecture}{Conjecture}
\newsharedtheorem{corollary}{Corollary}
\theoremstyle{definition}

\newsharedtheorem{definition}{Definition}
\newtheorem*{definition*}{Definition}
\newsharedtheorem{remark}{Remark}
\newtheorem*{remark*}{Remark}

\newrefformat{thm}{Theorem~\ref{#1}}
\newrefformat{cond}{Condition~\ref{#1}}
\newrefformat{cor}{Corollary~\ref{#1}}
\newrefformat{def}{Definition~\ref{#1}}
\newrefformat{definition}{Definition~\ref{#1}}
\newrefformat{lem}{Lemma~\ref{#1}}
\newrefformat{lemma}{Lemma~\ref{#1}}
\newrefformat{Alg}{Algorithm~\ref{#1}}
\newrefformat{observation}{Observation~\ref{#1}}
\newrefformat{claim}{Claim~\ref{#1}}
\newrefformat{example}{Example~\ref{#1}}
\newrefformat{fact}{Fact~\ref{#1}}
\newrefformat{prop}{Proposition~\ref{#1}}
\newrefformat{proposition}{Proposition~\ref{#1}}
\newrefformat{conjecture}{Conjecture~\ref{#1}}
\newrefformat{remark}{Remark~\ref{#1}}
\newrefformat{fig}{Figure~\ref{#1}}
\newrefformat{assumption}{Assumption~\ref{#1}}

\let\pref\prettyref
\newcommand{\Cref}[1]{\pref{#1}}

\def\*#1{\boldsymbol{#1}} 
\def\+#1{\mathcal{#1}} 
\def\-#1{\mathrm{#1}} 
\def\=#1{\mathbb{#1}} 
\def\!#1{\mathfrak{#1}} 

\newcommand{\bits}{\{0,1\}}
\newcommand{\defeq}{\triangleq}
\newcommand{\ASens}{\operatorname{ASens}}
\newcommand{\CSP}{\operatorname{CSP}}

\newcommand{\OR}{\operatorname{OR}}
\newcommand{\NAND}{\operatorname{NAND}}
\newcommand{\viol}{\operatorname{viol}}

\newcommand{\Unif}{\operatorname{Unif}}
\newcommand{\Law}{\operatorname{Law}}

\renewcommand{\P}{\ensuremath{\mathbf{P}}}
\newcommand{\NP}{\ensuremath{\mathbf{NP}}}

\DeclareMathOperator{\Sol}{Sol}
\DeclareMathOperator{\ExpDist}{Exp}
\DeclareMathOperator{\arity}{arity}
\newcommand{\var}[1]{\operatorname{var}(#1)}
\newcommand{\distH}{d_{\mathrm H}}
\newcommand{\Wass}{W_1}
\newcommand{\eps}{\varepsilon}

\DeclareMathOperator{\oPr}{\mathbf{Pr}}
\DeclareMathOperator{\opr}{\mathbb{P}}
\renewcommand{\Pr}[2][]
{\ifthenelse{\isempty{#1}}
  {\oPr\left[#2\right]}
  {\oPr_{#1}\left[#2\right]}
} 
\newcommand{\pr}[2][]
{\ifthenelse{\isempty{#1}}
  {\opr\left[#2\right]}
  {\opr_{#1}\left[#2\right]}
}

\DeclareMathOperator{\oE}{\mathbb{E}}
\newcommand{\E}[2][]
{\ifthenelse{\isempty{#1}}
  {\oE\left[#2\right]}
  {\oE_{#1}\left[#2\right]}
} 

\DeclareMathOperator{\oVar}{\mathrm{Var}}
\newcommand{\Var}[2][]
{\ifthenelse{\isempty{#1}}
  {\oVar\left[#2\right]}
  {\oVar_{#1}\left[#2\right]}
}
\def\oEnt{\mathrm{Ent}}
\NewDocumentCommand{\Ent}{ O{} O{} m }{
  \ifthenelse{\isempty{#1}} {
    \ifthenelse{\isempty{#2}} {
      \oEnt\left[#3\right]
    } {
      \oEnt^{#2}\left[#3\right]
    }
  } {
    \ifthenelse{\isempty{#2}} {
      \oEnt_{#1}\left[#3\right]
    } {
      \oEnt_{#1}^{#2}\left[#3\right]
    }
  }
}

\newcommand{\e}{\mathrm{e}}
\renewcommand{\emptyset}{\varnothing}

\newcommand{\tuple}[1]{\left(#1\right)}

\newcommand{\tp}{\tuple}

\newbool{doubleblind}
\setbool{doubleblind}{false}

\title{Stability Dichotomies for Boolean Constraint Satisfaction Problems}

\ifdoubleblind
\author{Author(s)}
\else
\author{Chunyang Wang, Yuichi Yoshida}
\address[Chunyang Wang, Yuichi Yoshida]{National Institute of Informatics, Tokyo, Japan. \textnormal{Email: \mbox{\texttt{\{c\_wang,yyoshida\}@nii.ac.jp}}}}

\fi

\begin{document}
\pagenumbering{roman}

\begin{abstract}
We study the \emph{stability} of Boolean constraint satisfaction problems
(CSPs) through the notion of \emph{average sensitivity} (Varma and Yoshida,
SODA 2021; SICOMP 2023).  It measures the expected $1$-Wasserstein distance
between an algorithm's output distributions before and after the deletion
of a uniformly chosen constraint, using the unnormalized Hamming metric.

We establish two dichotomies for every finite Boolean constraint language
$\Gamma$, where $n\geq 2$ denotes the number of variables in an instance.
For \emph{stable solvability}, exactly one of the following holds:
\begin{itemize}
    \item either there is an algorithm that solves $\CSP(\Gamma)$ and has
          average sensitivity $O_{\Gamma}(1)$ for all satisfiable instances;
    \item or every algorithm that solves $\CSP(\Gamma)$ has average sensitivity
          $\Omega_{\Gamma}(n)$ on satisfiable instances of arbitrarily large $n$.
\end{itemize}
The first alternative holds if and only if $\Gamma$ has \emph{finite duality}: unsatisfiability can be witnessed on a bounded number of variables.

For \emph{stable approximability}, where a $(1-\eps)$-approximation violates at most an $\eps$-fraction of the constraints in expectation, exactly one of the following holds:
\begin{itemize}
    \item either for every $\eps\in(0,1]$, there is an algorithm that
          $(1-\eps)$-approximates $\CSP(\Gamma)$ with average sensitivity
          $O_{\Gamma}(\eps^{-1}\log n)$ for all satisfiable instances;
    \item or there exists $\eps_{\Gamma}\in(0,1]$ such that
          every algorithm that $(1-\eps_{\Gamma})$-approximates
          $\CSP(\Gamma)$ has average sensitivity
          $\Omega_{\Gamma}(n)$ on satisfiable instances of arbitrarily large $n$.
\end{itemize}
The first alternative holds if and only if $\Gamma$ has \emph{bounded width}: local consistency checks on bounded sets of variables detect unsatisfiability.
\end{abstract}

\maketitle


\setcounter{tocdepth}{1}
\tableofcontents
\pagenumbering{arabic}

\section{Introduction}
\label{sec:introduction}

Constraint satisfaction problems (CSPs) provide a common framework for studying the computational complexity of problems such as Boolean satisfiability and graph coloring. A CSP asks whether values can be assigned to variables so that all constraints are satisfied. Boolean satisfiability is an NP-complete example~\cite{cook1971complexity}. More generally, the classification of CSPs asks how the complexity of deciding satisfiability depends on the constraint language, that is, the set of allowed constraint relations~\cite{schaefer1978complexity,FederVardi98}. The CSP Dichotomy Theorem states that, for every finite constraint language over a finite domain, this decision problem is either solvable in polynomial time or $\NP$-complete~\cite{bulatov2017dichotomy,zhuk2020dichotomy}.

The $\NP$-completeness of CSP satisfiability for some constraint languages motivates the study of approximation algorithms. Within this paradigm, a central objective is to characterize the \emph{efficient approximability of CSPs}: which CSPs can---or cannot---be efficiently approximated. The pursuit of this characterization has catalyzed major advances across theoretical computer science, including the PCP Theorem~\cite{uriel1996interactive,sanjeev1998proof,sanjeev1998probabilistic}, semidefinite programming (SDP) relaxations~\cite{geomans1995improved,karger1998approximate,zwick1998approximation,zwick1999outward,charikar2009near}, and the Unique Games Conjecture (UGC)~\cite{khot2002power,raghavendra2008optimal}. Assuming the UGC, Raghavendra~\cite{raghavendra2008optimal} established an optimal approximation characterization for almost-satisfiable CSP instances. However, in other regimes, including satisfiable CSP instances, efficient approximability is not fully characterized by the UGC and remains a highly active area of research~\cite{amey2002approximabilityI,amey2023approxmabilityII,amey2023approximabilityIII,amey2024approximabilityIV,amey2025approximabilityV}.

\medskip
\noindent
\textbf{Stable solvability and stable approximability of CSPs.}\,\, A rapidly growing body of literature in computer science studies algorithmic \emph{stability}, a fundamental concept that measures how an algorithm's output changes under perturbations to its input. In many real-world settings, merely computing a valid solution is insufficient, and the algorithm must also ensure that its output remains robust against small variations in the input data. This requirement has driven the development of various domain-specific stability frameworks. Notable examples include differential privacy~\cite{dwork2006calibrating,mcsherry2007mechanism,gupta2010differentially,bassily2021algorithmic} in data analysis, recourse in online algorithms~\cite{gupta2014maintaining,megow2016power,lattanzi2017consistent,dutting2025cost}, and various stability measures in machine learning~\cite{kearns1997algorithmic,bousquet2002stability,poggio2004general,elisseeff2005stability,mukherjee2006learning,shai2010learnability,hardt2016train,impagliazzo2022reproducibility}.


We study the stability of algorithms for CSPs through the lens of \emph{average sensitivity}, a general notion of algorithmic stability introduced by Varma and Yoshida~\cite{VarmaYoshida23}. This concept has attracted substantial attention and applies to a broad spectrum of problems~\cite{PengYoshida20,KumabeYoshida22,KumabeYoshidaKnapsack22,YoshidaIto22,HaraYoshida23,VarmaYoshida23,LiHeBaiPeng25,EbbensYoshida26,YoshidaZhang26}. 
In the context of CSPs, the average sensitivity of a randomized algorithm is measured by the expected $1$-Wasserstein distance between its original output distribution and the output distribution after deleting a uniformly chosen constraint of the input instance, using unnormalized Hamming distance as the transportation cost.

The first fundamental problem is whether we can \emph{stably} solve satisfiable CSPs---that is, solve them with low average sensitivity. This requirement proves surprisingly restrictive even for very simple instances. For example, every algorithm computing proper $2$-colorings has linear average sensitivity on some bipartite graphs~\cite{VarmaYoshida23}.\footnote{This lower bound applies even to computationally unbounded algorithms and requires no computational complexity assumption.} To date, existing literature on the average sensitivity of CSPs remains quite limited. It is therefore natural to ask:

\begin{center}
    \emph{Can we characterize which CSPs are stably solvable?}
\end{center}

Our second question concerns stable approximation. Even when a CSP is not stably solvable, allowing a small fraction of constraints to be violated may permit algorithms with low average sensitivity. Formally, we analyze the trade-off between the following two criteria:
\begin{itemize}
    \item \textbf{Approximation}: The (randomized) algorithm outputs assignments that satisfy a large fraction of the constraints in expectation.
    \item \textbf{Stability}: The algorithm exhibits low average sensitivity.
\end{itemize}

This trade-off between approximation and stability is highly natural and mirrors similar phenomena studied across various other contexts~\cite{duchi2014privacy,megow2016power,bassily2021algorithmic,mark2024stable,bender2026history}. We therefore ask:

\begin{center}
    \emph{Can we characterize which CSPs are stably approximable?}
\end{center}


\subsection{Our Results}
In this work, we provide a complete resolution of the two aforementioned questions for satisfiable Boolean CSPs parameterized by constraint languages. Specifically, we establish dichotomy theorems for both stable solvability and stable approximability.

Fix a finite Boolean constraint language $\Gamma$, defined as a finite family of relations over the Boolean domain $\bits$. Throughout, we will assume that $\Gamma$ is nonempty and every relation in $\Gamma$ is nonempty, and use $n\geq 2$ to denote the number of variables. Our algorithmic guarantees and sensitivity bounds apply to satisfiable input instances, and the algorithms are permitted to behave arbitrarily on unsatisfiable inputs.

Our first result characterizes stable solvability by \emph{finite duality}, the property that unsatisfiability can always be witnessed on a number of variables bounded solely in terms of $\Gamma$~\cite{larose2007characterization,BulatovKrokhinLarose08}.

\begin{theorem}[Dichotomy for stable solvability, informal]
\label{thm:intro-stable-solvability}
For every finite Boolean constraint language $\Gamma$, exactly one of the following holds:
\begin{itemize}
    \item either there exists an algorithm that solves $\CSP(\Gamma)$ with average sensitivity $O_{\Gamma}(1)$ for all satisfiable instances;
    \item or every algorithm that solves $\CSP(\Gamma)$ has average sensitivity $\Omega_{\Gamma}(n)$ on satisfiable instances with arbitrarily large $n$.
\end{itemize}
The first alternative holds if and only if $\Gamma$ has finite duality.
\end{theorem}

The formal statement of \Cref{thm:intro-stable-solvability} is provided in \Cref{thm:stable-solvability-dichotomy}.

Our second result characterizes stable approximability by \emph{bounded width}: unsatisfiability can always be detected by a local consistency algorithm that examines only a bounded number of variables at a time~\cite{BartoKozik09CD,BartoKozik14}.

\begin{theorem}[Dichotomy for stable approximability, informal]
\label{thm:intro-stable-approximability}
For every finite Boolean constraint language $\Gamma$, exactly one of the following holds:
\begin{itemize}
    \item either for every $\eps\in (0,1]$, there exists an algorithm that $(1-\eps)$-approximates $\CSP(\Gamma)$ with average sensitivity $O_{\Gamma}(\eps^{-1}\log n)$ for all satisfiable instances;
    \item or there exists an $\eps_{\Gamma}\in(0,1]$ such that every algorithm that $(1-\eps_{\Gamma})$-approximates $\CSP(\Gamma)$ has average sensitivity $\Omega_{\Gamma}(n)$  on satisfiable instances with arbitrarily large $n$.
\end{itemize}
The first alternative holds if and only if $\Gamma$ has bounded width.
\end{theorem}

The formal statement of \Cref{thm:intro-stable-approximability} is provided in \Cref{thm:stable-approximability-dichotomy}.

\begin{remark}[Computational complexity]\label{rem:computational-complexity}
    We emphasize that the classifications in Theorems~\ref{thm:intro-stable-solvability} and \ref{thm:intro-stable-approximability} are purely information-theoretic. Consequently, our average-sensitivity lower bounds are unconditional and do not rely on unproven computational complexity assumptions (such as $\P\neq\NP$). Furthermore, while the mathematical definition of stability permits computationally unbounded algorithms, the specific low-sensitivity algorithms we design use polynomially many operations in the model where sampling uniform and exponential random variables and performing arithmetic and comparisons on them take constant time.
\end{remark}

We remark that finite duality is a strictly more restrictive structural condition than bounded width. Consequently, \Cref{thm:intro-stable-solvability} and \Cref{thm:intro-stable-approximability} jointly establish a stability trichotomy for satisfiable Boolean CSPs:
\begin{itemize}
    \item $\CSP(\Gamma)$ is stably solvable if $\Gamma$ has finite duality;
    \item $\CSP(\Gamma)$ is stably approximable, but not stably solvable, if $\Gamma$ has bounded width but lacks finite duality;
    \item $\CSP(\Gamma)$ is not stably approximable if $\Gamma$ lacks bounded width: for some fixed accuracy $1-\eps_\Gamma$, every approximation algorithm has linear average sensitivity on arbitrarily large satisfiable instances.
\end{itemize}

\subsection{Technical Overview}

Our strategy for proving the dichotomies for stable solvability and stable approximability (Theorems \ref{thm:intro-stable-solvability} and \ref{thm:intro-stable-approximability}) proceeds as follows:
\begin{itemize}
    \item \textbf{Upper bounds:} To establish both stability upper bounds, we use the relational characterizations of finite duality and bounded width. These characterizations cover the CSPs by several, possibly overlapping, structural cases. We then design new algorithms with low average sensitivity for each nontrivial case.
    
    \item \textbf{Lower bound for stable solvability:} We construct a family of satisfiable instances whose solutions must be far apart, but any two instances in the family yield the same subinstance after deleting one constraint from each. This converts solution separation into an average-sensitivity lower bound. The family comes from adapting the construction for languages without finite duality in~\cite{larose2007characterization}.
    
    \item \textbf{Lower bound for stable approximability:} We transfer lower bounds between constraint languages by replacing each constraint of one language with a constant-size conjunction of constraints of another language, using the original variables and new existential variables. Such replacements are described by \emph{primitive-positive (pp) definitions}. Our reduction then transfers an average-sensitivity lower bound obtained from the locally testable codes of Yoshida and Zhang~\cite{YoshidaZhang26}.
\end{itemize}

We next provide more in-depth explanations for the algorithms and the mechanisms behind the lower bounds.

\subsubsection{Stable Solvability: Upper Bound}
To establish the upper bound of the dichotomy for stable solvability, we utilize the relational characterization that a constraint language has finite duality if and only if it is $0$-valid, $1$-valid, strictly essentially positive, or strictly essentially negative.

A language is $0$-valid (resp., $1$-valid) when the all-zero (resp., all-one) tuple belongs to each relation, so a fixed assignment gives zero average sensitivity. A strictly essentially positive relation has a CNF definition using only unit clauses and positive clauses. On a satisfiable instance, assign $0$ to variables occurring in negative unit clauses of these definitions and $1$ to all other variables. This satisfies every constraint, and deleting a constraint can change only variables in its scope. Strict essential negativity is the dual case, with positive unit clauses forcing $1$ and all remaining variables assigned $0$.

\subsubsection{Stable Solvability: Lower Bound}

To prove the average-sensitivity lower bound for algorithms that solve every satisfiable instance, we construct sparse \emph{minimally unsatisfiable} instances with large incidence-graph diameter. Minimal unsatisfiability means that deleting any single constraint makes the instance satisfiable; the incidence graph has a node for each variable and constraint, with edges recording occurrences. Finite duality is equivalent to a uniform bound on the size of such obstructions after isolated variables are removed.

Solutions obtained by deleting two distinct constraints must disagree on all variable nodes of some path connecting the deleted constraints. Otherwise, the two assignments could be combined into a solution of the unsatisfiable instance. Thus, distant deletions force distant output distributions. Comparing them through the common instance with both constraints deleted and averaging along a long shortest path gives a linear sensitivity lower bound when the number of constraints is linear in the diameter. We obtain the required instances by adapting the construction of Larose et al.~\cite{larose2007characterization}.

\subsubsection{Stable Approximability: Upper Bound}
\label{sec:stable-approximability-upper-overview}
Fix $\eps\in(0,1]$. The Boolean bounded-width criterion reduces the upper bound to $0$-valid, $1$-valid, Horn, dual-Horn, and bijunctive languages. A Horn relation has a CNF definition with at most one positive literal per clause; a dual-Horn relation has at most one negative literal per clause; a bijunctive relation has a $2$-CNF definition. The valid cases use fixed assignments, so we focus on the remaining cases.

\paragraph{\textbf{Horn and Dual-Horn Languages.}}
It is well known that satisfiable Horn (resp., dual-Horn) instances can be solved by taking the coordinatewise least (resp., greatest) solution. However, the coordinatewise least (resp., greatest) solution of a Horn (resp., dual-Horn) instance can be highly unstable under uniform constraint deletion. A counterexample is the chain instance on $n$ variables $x_0,x_1,\dots,x_{n-1}$, consisting of one unit constraint $x_0=1$ and $n-1$ implication constraints $x_i\implies x_{i+1}$ for $0\le i<n-1$.

For an instance $\+I=(V,\+C)$, where $V$ is its variable set and $\+C$ its constraint multiset, retaining each constraint independently with a fixed probability $p\in[1-\eps,1]$ guarantees the desired approximation, but does not by itself bound sensitivity for that value of $p$. Instead, we draw $p\sim\Unif([1-\eps,1])$ and output the least solution of the retained instance. As $p$ increases, this solution is coordinatewise nondecreasing. Integrating the total influence of the constraints over $p$ bounds the total expected change by the number of non-isolated variables, at most $r_\Gamma|\+C|$, where $r_\Gamma$ bounds the relation arities. Dividing by the interval length $\eps$ and the number of constraints gives average sensitivity at most $r_\Gamma/\eps$. The dual-Horn case uses the greatest solution.

\paragraph{\textbf{A Warm-up: $2$-Coloring.}}

As a warm-up for general bijunctive languages, we first design a stable approximation algorithm for $2$-coloring. The $2$-coloring problem is a Boolean CSP over the bijunctive language consisting only of the disequality relation $\mathrm{NEQ}\defeq\{(0,1),(1,0)\}$.

For a bipartite graph $G=(V,E)$, use the exponential-shift low-diameter decomposition of Miller, Peng, and Xu~\cite{miller2013parallel}. Independently sample $\delta_v\sim\ExpDist(\beta)$, an exponential random variable of rate $\beta=\eps$, for each vertex $v$. The root $r(u)$ maximizes $\delta_v-\mathrm{dist}_G(u,v)$, with ties broken by a fixed ordering. These roots define clusters, each of which admits a shortest-path tree from its root. Color each root $0$ and alternate colors along its tree. Bipartiteness ensures that every edge within a cluster, including non-tree edges, is satisfied (\Cref{fig:ldd-two-coloring}).

The probability that an edge joins different clusters is at most $\beta=\eps$, giving the approximation guarantee. Couple the shifts before and after an edge deletion. A vertex $v$ can change its color only if the deleted edge lies on a chosen shortest path from $r(v)$ to $v$. The expected path length is $O(\beta^{-1}\log n)$. Summing over the at most $2|E|$ non-isolated vertices and dividing by $|E|$ gives the same bound on average sensitivity.

\begin{figure}[htbp]
    \centering
\begin{tikzpicture}[
  scale=1.08,transform shape,
  x=1.25cm,y=1.1cm,
  every node/.style={font=\small},
  vertex/.style={circle,draw=black!70,fill=white,minimum size=5.5mm,inner sep=0pt},
  root/.style={double,double distance=.8pt},
  zero/.style={vertex,draw=black,fill=white},
  one/.style={vertex,draw=black,fill=black!80,text=white},
  edge/.style={draw=black!65,line width=.7pt},
  rcluster/.style={draw=blue!55!black,fill=blue!10,line width=.65pt,rounded corners=4pt},
  scluster/.style={draw=orange!70!black,fill=orange!13,line width=.65pt,rounded corners=4pt},
  tcluster/.style={draw=teal!65!black,fill=teal!11,line width=.65pt,rounded corners=4pt}
]
  \foreach \panel/\offset in {0/0,1/6.7} {
    \begin{scope}[xshift=\offset cm]
      \ifnum\panel=0
        \node[font=\small\bfseries] at (1.5,3.25) {Before deleting $e$};
        \path[rcluster]
          (-.43,.57)--(.43,.57)--(.43,1.57)--(1.43,1.57)
          --(1.43,2.43)--(-.43,2.43)--cycle;
        \path[scluster]
          (1.57,.57)--(2.57,.57)--(2.57,-.43)--(3.43,-.43)
          --(3.43,2.43)--(1.57,2.43)--cycle;
      \else
        \node[font=\small\bfseries] at (1.5,3.25) {After deleting $e$};
        \path[rcluster] (-.43,.57) rectangle (.43,2.43);
        \path[scluster]
          (1.57,.57)--(2.57,.57)--(2.57,-.43)--(3.43,-.43)
          --(3.43,2.43)--(.57,2.43)--(.57,1.57)--(1.57,1.57)--cycle;
      \fi
      \path[tcluster]
        (-.43,-.43)--(2.43,-.43)--(2.43,.43)--(1.43,.43)
        --(1.43,1.43)--(.57,1.43)--(.57,.43)--(-.43,.43)--cycle;

      \foreach \y in {0,1} {\draw[edge] (0,\y)--(3,\y);}
      \draw[edge] (1,2)--(3,2);
      \foreach \x in {0,1,2,3} {\draw[edge] (\x,0)--(\x,2);}
      \ifnum\panel=0
        \draw[edge] (0,2)--(1,2);
        \node[font=\scriptsize] at (.5,2.22) {$e$};
      \else
        \draw[red!75!black,densely dashed,line width=.8pt] (0,2)--(1,2);
        \draw[red!75!black,line width=1pt] (.43,1.92)--(.57,2.08);
        \draw[red!75!black,line width=1pt] (.43,2.08)--(.57,1.92);
      \fi

      \node[zero,root] at (0,2) {0};
      \node at (1,2.55) {$v$};
      \ifnum\panel=0
        \node[one] at (1,2) {1};
      \else
        \node[zero] at (1,2) {0};
      \fi
      \node[one] at (2,2) {1};
      \node[zero,root] at (3,2) {0};
      \node[one] at (0,1) {1};
      \node[one] at (1,1) {1};
      \node[zero] at (2,1) {0};
      \node[one] at (3,1) {1};
      \node[one] at (0,0) {1};
      \node[zero,root] at (1,0) {0};
      \node[one] at (2,0) {1};
      \node[zero] at (3,0) {0};

      \node[text=blue!55!black] at (0,2.7) {$r$, $\delta_r=4.6$};
      \node[text=orange!70!black] at (3,2.7) {$s$, $\delta_s=4.4$};
      \node[text=teal!65!black] at (1,-.72) {$t$, $\delta_t=4.1$};
    \end{scope}
  }
  \node[font=\footnotesize,transform shape=false] at (4.18,-1.3)
    {The exponential shifts are identical in both instances; every unlabelled shift equals $0.1$};
\end{tikzpicture}
    \caption{An illustration of exponential-shift clustering on a bipartite
    graph before (left) and after (right) deleting the edge $e$. Black and
    white nodes indicate output colors $1$ and $0$, respectively. Colored
    regions group vertices by their selected roots, and double circles mark
    these roots. The exponential shifts are identical in the two panels.
    Deleting $e$ changes the selected root of $v$ from $r$ to $s$ and its
    output color from $1$ to $0$.}
    \label{fig:ldd-two-coloring}
\end{figure}

\paragraph{\textbf{From $2$-Coloring to $2$-SAT.}}
For $2$-SAT, we propagate scores along directed paths of implications between literals. A candidate root must be a \emph{feasible literal}, meaning that it is true in some satisfying assignment. Equivalently, in a satisfiable instance there is no directed path from that literal to its negation. Deleting a constraint can make new roots feasible, creating a second source of sensitivity in addition to changes in path lengths (\Cref{fig:ldd-two-sat}).

Assign independent uniform priorities in $[0,1]$ to constraints and retain those with priority at most $1-\eps/2$ to form the implication graph used to compute distances. A second, independent threshold, uniform on $[1-\eps/2,1]$, gives a larger subinstance used to test root feasibility. Give each literal an independent exponential shift of rate $\beta=\Theta_\Gamma(\eps)$, with the precise choice specified in \Cref{def:alg-bijunctive}. A literal's score is the maximum, over feasible roots, of the root's shift minus its directed distance to that literal in the implication graph of the smaller subinstance. Set each variable to $1$ when its positive literal has the higher score, and to $0$ otherwise. Path changes are bounded by expected path lengths as in $2$-coloring. For feasibility changes, the exponential shifts ensure that, as constraints are added to the subinstance used to test root feasibility, the expected number of times the root determining a variable's value changes is $O(\log n)$. Averaging over the second threshold yields a total sensitivity bound of $O_\Gamma(\eps^{-1}\log n)$.

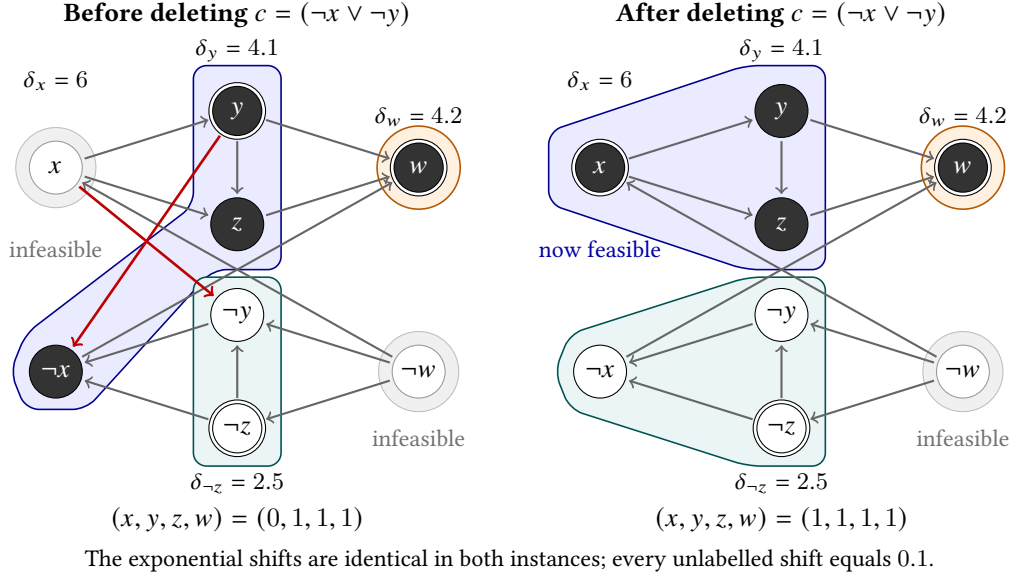
\begin{figure}[htbp]
    \centering
\begin{tikzpicture}[
  x=1cm,y=1cm,
  every node/.style={font=\small},
  literal/.style={circle,draw=black,fill=white,
    minimum size=7mm,inner sep=0pt},
  true literal/.style={literal,fill=black!80,text=white},
  selected root/.style={double,double distance=1pt},
  arc/.style={->,draw=black!60,line width=.8pt,
    shorten <=4mm,shorten >=4mm},
  cluster/.style={line width=.6pt,rounded corners=5pt},
  annotation/.style={font=\footnotesize}
]
  \foreach \panel/\offset in {0/0,1/7.2} {
    \begin{scope}[xshift=\offset cm]
      \ifnum\panel=0
        \node[font=\small\bfseries] at (2.4,4.62)
          {Before deleting $c=(\neg x\vee\neg y)$};
        \path[cluster,draw=blue!55!black,fill=blue!8]
          (-.6,-.1)--(-.4,.4)--(1.82,2.1)--(1.82,3.95)
          --(2.98,3.95)--(2.98,1.25)--(2.08,1.25)
          --(.4,-.6)--(-.4,-.6)--cycle;
        \path[cluster,draw=teal!65!black,fill=teal!8]
          (1.82,-1.35) rectangle (2.98,1.15);
        \path[draw=black!25,fill=black!6] (0,2.6) circle (.53);
      \else
        \node[font=\small\bfseries] at (2.4,4.62)
          {After deleting $c=(\neg x\vee\neg y)$};
        \path[cluster,draw=blue!55!black,fill=blue!8]
          (-.68,2.12)--(-.68,3.08)--(1.92,3.95)--(2.98,3.95)
          --(2.98,1.25)--(1.92,1.25)--cycle;
        \path[cluster,draw=teal!65!black,fill=teal!8]
          (-.6,-.1)--(-.4,.4)--(1.95,1.15)--(2.98,1.15)
          --(2.98,-1.35)--(1.95,-1.35)--(-.4,-.6)--cycle;
      \fi
      \path[draw=orange!70!black,fill=orange!12,line width=.6pt]
        (4.8,2.6) circle (.55);
      \path[draw=black!25,fill=black!6] (4.8,-.1) circle (.53);

      \draw[arc] (0,2.6)--(2.4,3.35);
      \draw[arc] (0,2.6)--(2.4,1.85);
      \draw[arc] (2.4,3.35)--(2.4,1.85);
      \draw[arc] (2.4,3.35)--(4.8,2.6);
      \draw[arc] (2.4,1.85)--(4.8,2.6);
      \draw[arc] (2.4,.65)--(0,-.1);
      \draw[arc] (2.4,-.85)--(0,-.1);
      \draw[arc] (2.4,-.85)--(2.4,.65);
      \draw[arc] (4.8,-.1)--(2.4,.65);
      \draw[arc] (4.8,-.1)--(2.4,-.85);
      \draw[arc] (0,-.1)--(4.8,2.6);
      \draw[arc] (4.8,-.1)--(0,2.6);

      \ifnum\panel=0
        \draw[arc,draw=red!75!black,line width=1pt] (0,2.6)--(2.4,.65);
        \draw[arc,draw=red!75!black,line width=1pt] (2.4,3.35)--(0,-.1);
        \node[literal,draw=black!40] at (0,2.6) {$x$};
        \node[true literal,selected root] at (2.4,3.35) {$y$};
        \node[true literal] at (0,-.1) {$\neg x$};
        \node[annotation,text=black!55] at (0,1.53) {infeasible};
      \else
        \node[true literal,selected root] at (0,2.6) {$x$};
        \node[true literal] at (2.4,3.35) {$y$};
        \node[literal] at (0,-.1) {$\neg x$};
        \node[annotation,text=blue!65!black] at (0,1.53) {now feasible};
      \fi
      \node[true literal] at (2.4,1.85) {$z$};
      \node[true literal,selected root] at (4.8,2.6) {$w$};
      \node[literal] at (2.4,.65) {$\neg y$};
      \node[literal,selected root] at (2.4,-.85) {$\neg z$};
      \node[literal,draw=black!40] at (4.8,-.1) {$\neg w$};
      \node[annotation,text=black!55] at (4.8,-.98) {infeasible};

      \node[annotation] at (0,3.75) {$\delta_x=6$};
      \node[annotation] at (2.4,4.15) {$\delta_y=4.1$};
      \node[annotation] at (4.8,3.27) {$\delta_w=4.2$};
      \node[annotation] at (2.4,-1.58) {$\delta_{\neg z}=2.5$};
      \ifnum\panel=0
        \node at (2.4,-2.05) {$(x,y,z,w)=(0,1,1,1)$};
      \else
        \node at (2.4,-2.05) {$(x,y,z,w)=(1,1,1,1)$};
      \fi
    \end{scope}
  }
  \node[annotation] at (6,-2.6)
    {The exponential shifts are identical in both instances; every unlabelled shift equals $0.1$.};
\end{tikzpicture}
    \caption{A simplified illustration of literal scores on the full implication graph of
    $(\neg x\vee y)\wedge(\neg x\vee z)\wedge(\neg y\vee z)
    \wedge(\neg y\vee w)\wedge(\neg z\vee w)\wedge(x\vee w)\wedge c$,
    where $c=(\neg x\vee\neg y)$, before (left) and after (right) deleting $c$.
    Black and white nodes indicate true and false literals, respectively. Colored regions group literals according to the roots attaining their scores, and double circles mark these roots; gray regions mark infeasible literals. The root attaining a literal's score need not be the root that determines the value of its variable (as with $\neg z$). The newly feasible literal $x$ changes the output from $(0,1,1,1)$ to $(1,1,1,1)$. The nested random subinstances of the actual algorithm (\Cref{def:alg-bijunctive}) are omitted here.}
    \label{fig:ldd-two-sat}
\end{figure}

\subsubsection{Stable Approximability: Lower Bound}
To prove the lower bound for stable approximability, we develop a reduction that transfers lower bounds on average sensitivity between constraint languages related by primitive-positive (pp) definitions.
Earlier sensitivity-preserving reductions addressed worst-case sensitivity via the PCP framework~\cite{FlemingYoshida26}.

For average sensitivity, deleting one constraint before the reduction may remove several constraints after the reduction, whereas the available guarantee averages over single deletions. We therefore delete each replacement constraint independently with a fixed probability. This makes the successive deletions arising from one original constraint comparable to conditioned samples of the same random instance. For the instances used in the lower bound, we reserve the same collection of disjoint sets of new existential variables before and after deleting an original constraint, and assign distinct sets to distinct constraints uniformly at random. We couple the two reductions so that unchanged original constraints receive the same existential variables and the same choices of replacement constraints to retain. The reduced instances then have the same variable set and differ only by deletion of the retained replacement constraints generated by the deleted original constraint.

Finally, we adapt the locally testable code (LTC) construction of Yoshida and Zhang~\cite{YoshidaZhang26} to obtain an average-sensitivity lower bound for a CSP with four-variable parity constraints. Transferring this lower bound through the reduction completes the proof.

\subsection{Related Work}

\subsubsection{Classification of CSPs by Constraint Languages}

There is an extensive literature on the classification of CSPs by their constraint languages. Such classifications have been studied for satisfiability~\cite{schaefer1978complexity,FederVardi98,bulatov2017dichotomy,zhuk2020dichotomy}, optimization~\cite{nadia1995dichotomy,sanjeev2001approximability,raghavendra2008optimal,thapper2016complexity,kolmogorov2017complexity}, exact counting~\cite{bulatov2013complexity,dyer2013effective,cai2017complexity}, approximate counting~\cite{dyer2010approximation}, property testing~\cite{bhattacharyya2013algebraic,chen2019constant} and robust satisfiability~\cite{dalmau2013robust,BartoKozik16}. For Boolean languages, our characterization of stable approximability by bounded width coincides, assuming $\P\neq\NP$, with the characterization of polynomial-time robust solvability. Here, robust solvability means that a polynomial-time algorithm, given an almost satisfiable instance, finds an assignment whose value approaches one as the optimum approaches one.

\subsubsection{Logical Expressibility and Reductions between CSPs}
Our lower bound also compares the expressive power of constraint languages through primitive-positive (pp) definitions, a standard formalism for expressing the relations of one constraint language in terms of another. A pp definition expresses a relation as the projection of a conjunction of relations from the base language. Over finite domains, pp definability is characterized exactly by polymorphism preservation~\cite{jeavons1997closure,bulatov2005classifying,
barto2017polymorphism}. See~\cite{barto2017polymorphism} for related notions of pp interpretations and pp constructions for more general gadget reductions.

Our reduction transfers linear lower bounds on average sensitivity between constraint languages. Previous work studies reductions that preserve approximation quality~\cite{papadimitriou1991optimization,sanjeev2001approximability,martin2000relative}, fixed-parameter tractability~\cite{downey1995fixed}, fine-grained complexity~\cite{dell2021finegrained}, and worst-case sensitivity~\cite{FlemingYoshida26}.

\subsubsection{Stability and Solution Space Geometry}
Stable solvability has a geometric interpretation. Each satisfiable instance determines a set of solutions in Hamming space, and a randomized algorithm selects a distribution supported on this set. The question is whether these distributions can be chosen across all satisfiable instances of a given CSP so that, for each instance, the expected Wasserstein distance before and after a uniformly random constraint deletion is small. Work on random CSPs studies the geometry within a single solution space and its connection to \emph{computational phase transitions}~\cite{mezard2002analytic,mezard2005clustering,achlioptas2006solution,florent2007gibbs,achlioptas2008algorithmic,montanari2011reconstruction}.

\emph{Coupling independence}~\cite{CZ23} and average sensitivity both use the $1$-Wasserstein distance to compare distributions under local perturbations. Coupling independence bounds the distance between conditional Gibbs distributions under a change in one boundary spin. These distributions are specified by the model and the boundary conditions, whereas stable solvability allows the algorithm to choose distributions supported on solutions to achieve low sensitivity. Bounds on the influence of boundary spins have been used to obtain efficient approximate counting and sampling algorithms~\cite{chen2024fast2,wang2024sampling,he2025FPTAS,chen2025deterministic,chen2025counting}.

\section{Preliminaries and Notation}
\label{sec:preliminaries}

\subsection{Boolean Constraint Satisfaction Problems}
\label{subsec:boolean-csps}

We begin by formally defining Boolean constraint satisfaction problems (CSPs). A Boolean CSP instance is a pair $\+I = (V, \+C)$, where $V$ is a finite set of variables and $\+C$ is a finite multiset of constraints. Each constraint is a pair $c = (\mathbf{s}, R)$, where the \emph{scope} $\mathbf{s}$ is a $k$-tuple of variables drawn from $V$, and the \emph{constraint relation} $R \subseteq \{0,1\}^k$ is a $k$-ary Boolean relation. We write $\mathbf{s}_c$ for the scope of $c$ and $\var{c}$ for the set of distinct variables in this scope, where repeated variables in the scope are allowed. Note that each member of the multiset $\+C$ is a separate constraint, and constraints with the same scope and relation are treated separately. 

An \emph{assignment} for $\+I$ is a Boolean mapping $\sigma : V \to \{0,1\}$. We say that $\sigma$ \emph{satisfies} a constraint $c = (\mathbf{s}, R)$ if $\sigma(\mathbf{s}) \in R$. Furthermore, $\sigma$ is called a \emph{satisfying assignment} (or \emph{solution}) of $\+I$ if it satisfies every constraint in $\+C$. We denote the set of all solutions to $\+I$ by $\Sol(\+I)$. We call
$\+I$ \emph{satisfiable} when this set is nonempty and \emph{unsatisfiable} otherwise.

We often study Boolean CSPs parameterized by finite \emph{constraint languages}. A Boolean constraint language $\Gamma$ is a finite set of Boolean relations, each of finite positive arity. We assume throughout that $\Gamma$ is nonempty and every relation in $\Gamma$ is nonempty. Given a language $\Gamma$, we use $\mathrm{CSP}(\Gamma)$ to denote the class of all CSP instances whose constraint relations are drawn from $\Gamma$. We define
$r_\Gamma \defeq \max\bigl(\{\arity(R) : R \in \Gamma\}\bigr)$
as an upper bound on the relation arities.

Given a Boolean CSP instance $\+I=(V,\+C)$,
\begin{itemize}
    \item for any subset of constraints $\+E\subseteq \+C$, define
          $\+I[\+E]\defeq(V,\+E)$, the subinstance containing exactly the
          constraints in $\+E$;
    \item for any subset of variables $U\subseteq V$, define
          \[
            \+I[U]
            \defeq
            \bigl(U,\{c\in\+C:\var{c}\subseteq U\}\bigr),
          \]
          the subinstance induced by $U$.
\end{itemize}

As single-constraint deletions recur throughout, for a constraint $c\in\+C$ we adopt the simplified notation $\+I-c\defeq\+I[\+C\setminus\{c\}]$ for the Boolean CSP instance obtained by deleting $c$.

For a Boolean CSP instance $\+I=(V,\+C)$ and an assignment $x\in\bits^V$, let $\viol_{\+I}(x)$ be the number of
constraints of $\+I$ violated by $x$.

\begin{definition}[Solving and approximating]
Let $\+A$ be a randomized algorithm that maps each Boolean CSP instance
$\+I=(V,\+C)$ to a random assignment $\+A(\+I)\in\bits^V$.  We say that
$\+A$ \emph{solves} a satisfiable instance $\+I$ if
\[
  \Pr{\+A(\+I)\in\Sol(\+I)}=1.
\]
For a satisfiable instance $\+I$ and $\eps\in[0,1]$, we say that $\+A$
\emph{$(1-\eps)$-approximates} $\+I$ if
\begin{equation}
  \E{\viol_{\+I}(\+A(\+I))}\leq \eps|\+C|.
  \label{eq:approx-guarantee}
\end{equation}
Unless additional random choices are specified, probabilities and expectations involving $\+A$ are over its internal randomness. Algorithms specified only for satisfiable inputs are extended to all inputs by returning an arbitrary assignment on unsatisfiable inputs.

We say that the algorithm $\+A$ \emph{solves $\CSP(\Gamma)$} (respectively,
\emph{$(1-\eps)$-approximates $\CSP(\Gamma)$}) if the corresponding
condition holds for every satisfiable instance of $\CSP(\Gamma)$.  No
condition is imposed on the output of $\+A$ on
unsatisfiable instances.
\end{definition}
\subsection{Boolean Relation Classes and Primitive-Positive Definitions}
\label{subsec:boolean-relation-classes}
We describe the relation classes that determine our upper bounds, then recall how pp definitions let us express the constraints used in the lower bound.

Let $R\subseteq\bits^k$. A CNF definition of $R$ is a CNF formula $\phi(x_1,\ldots,x_k)$ whose satisfying assignments are exactly the tuples in $R$. A literal is \emph{positive} if it is an unnegated variable and \emph{negative} otherwise. A clause is \emph{positive} (resp., \emph{negative}) if all of its literals are positive (resp., negative).

\begin{definition}[Boolean relation classes]
A Boolean relation $R$ is \emph{$0$-valid} if $(0,\ldots,0)\in R$, and it is \emph{$1$-valid} if $(1,\ldots,1)\in R$. We further say that $R$ is
\begin{itemize}[leftmargin=2.2em]
  \item \emph{Horn} if it has a CNF definition in which every clause contains at most one positive literal;
  \item \emph{dual-Horn} if it has a CNF definition in which every clause contains at most one negative literal;
  \item \emph{bijunctive} if it has a $2$-CNF definition;
  \item \emph{affine} if it is the solution set of a system of linear equations over $\=F_2$.
\end{itemize}
A Boolean constraint language $\Gamma$ has any of these properties if \emph{every} relation $R\in\Gamma$ has that property.
\end{definition}

We remark that the binary equality relation $x=y \iff (\neg x\lor y)\land (x\lor \neg y)$ is Horn, dual-Horn, bijunctive, and affine.

Schaefer's dichotomy theorem has the following standard relational
formulation~\cite{schaefer1978complexity}.

\begin{theorem}[Schaefer's dichotomy theorem]
\label{thm:schaefer-dichotomy}
For every finite Boolean constraint language $\Gamma$, the decision problem
$\CSP(\Gamma)$ is solvable in polynomial time if $\Gamma$ is $0$-valid,
$1$-valid, Horn, dual-Horn, bijunctive, or affine, and is $\NP$-complete
otherwise.
\end{theorem}

\subsubsection{Finite Duality}
\label{subsec:finite-duality}

The property of finite duality for a Boolean constraint language says that unsatisfiability can always be certified by examining a bounded number of variables. We formalize this as follows.

\begin{definition}[Finite duality]\label{def:finite-duality}
A constraint language $\Gamma$ has \emph{finite duality} if there exists a nonnegative integer $d_\Gamma$ such that every unsatisfiable instance $\+I \in \mathrm{CSP}(\Gamma)$ contains a subset of variables $U \subseteq V(\+I)$ with $|U| \leq d_\Gamma$ for which the induced subinstance $\+I[U]$ is unsatisfiable.
\end{definition}

The standard definition of finite duality uses finite obstruction sets for relational structures~\cite{larose2007characterization,BulatovKrokhinLarose08}, which is equivalent to \Cref{def:finite-duality} in the context of Boolean constraint languages. Furthermore, a constraint language $\Gamma$ has finite duality if and only if $\mathrm{CSP}(\Gamma)$ is first-order definable~\cite{rossman2005existential,atserias2008digraph}. We refer the reader to~\cite{BulatovKrokhinLarose08} for more equivalent characterizations of finite duality.

We provide a convenient equivalent relational characterization of finite duality for the purposes of this paper. The following definition is from~\cite{MahmoodMeierSchmidt21}.

\begin{definition}[Strict essential positivity and negativity]\label{def:strict-essential-positivity}
A Boolean relation $R$ is \emph{strictly essentially positive} (resp., \emph{strictly essentially negative}) if it has a CNF definition consisting exclusively of unit clauses and positive clauses (resp., negative clauses).\footnote{This differs from \emph{essential positivity} (resp., \emph{essential negativity}), which additionally allows the equality relation. These strict classes are also called $\OR$-conj and $\NAND$-conj, respectively: conjunctions of unit clauses (pins) and positive or negative clauses~\cite{dyer2012complexity}.}
\end{definition}

Recall that a Boolean constraint language $\Gamma$ is said to possess any of these properties if every relation $R \in \Gamma$ possesses that property. The following relational formulation specializes the Boolean finite-duality criterion in~\cite[Example~20]{BulatovKrokhinLarose08}.

\begin{proposition}[Boolean finite-duality criterion]
\label{prop:boolean-finite-duality-characterization}
A Boolean constraint language $\Gamma$ has finite duality if and only
if $\Gamma$ is $0$-valid, $1$-valid, strictly essentially positive or strictly essentially negative.
\end{proposition}

\subsubsection{Bounded Width}
\label{subsec:bounded-width}

The property of bounded width for a Boolean constraint language says that unsatisfiability can be certified via local propagation. We now state the canonical definition of bounded width, which is formulated in terms of solvability by local consistency algorithms~\cite{BartoKozik09CD,BartoKozik14}. We note that this definition is provided solely for completeness, as we do not directly employ this algorithmic characterization in our subsequent proofs.

\begin{definition}[Local consistency algorithms]
Let $\+I=(V,\+C)$ be a Boolean CSP instance. A partial assignment $h: U \to \{0,1\}$ defined on a subset of variables $U \subseteq V$ is \emph{locally satisfying} if
\[
  h(\mathbf{s}) \in R \qquad \text{for every constraint } c = (\mathbf{s}, R) \in \+C \text{ such that } \var{c} \subseteq U.
\]

Let $k$ and $\ell$ be integers such that $1 \leq k \leq \ell$. A \emph{$(k,\ell)$-strategy} for $\+I$ is a nonempty set $\+H$ of locally satisfying partial assignments such that:
\begin{enumerate}[label=(\roman*),leftmargin=2.2em]
  \item $|\mathrm{dom}(h)| \leq \ell$ for every $h \in \+H$;
  \item if $h \in \+H$ and $T \subseteq \mathrm{dom}(h)$, then $h_T \in \+H$, where $h_T$ denotes the restriction of the partial assignment $h$ to the subset $T$;
  \item if $S \subseteq T \subseteq V$ with $|S| \leq k$ and $|T| \leq \ell$, and $h \in \+H$ is a partial assignment with $\mathrm{dom}(h) = S$, then there exists an extension $g \in \+H$ such that $\mathrm{dom}(g) = T$ and $g_S = h$.
\end{enumerate}
\end{definition}

\begin{definition}[Bounded width]\label{def:bounded-width}
For fixed $1\leq k\leq \ell$, a Boolean language $\Gamma$ has \emph{width $(k,\ell)$} if every Boolean CSP
instance in $\CSP(\Gamma)$ that admits a nonempty $(k,\ell)$-strategy is
satisfiable.  

A Boolean language $\Gamma$ has \emph{bounded width} if it has width $(k,\ell)$ for
some fixed $1\leq k\leq\ell$. 
\end{definition}

Bounded width has several equivalent characterizations. It is equivalent to
the existence of a fixed Datalog program that recognizes exactly the
unsatisfiable instances~\cite{FederVardi98}, and to having bounded treewidth
duality~\cite{BulatovKrokhinLarose08}. Equivalently, unsatisfiability is
definable by an existential-positive infinitary sentence using finitely many
variables~\cite{kolaitis2000conjunctive,BulatovKrokhinLarose08}. In the
existential pebble-game formulation, there is a fixed pebble bound such that
the second player can respond to every move indefinitely if and only if the
instance is satisfiable~\cite{FederVardi98,kolaitis2000conjunctive}. Bounded width also
implies polynomial-time robust solvability~\cite{BartoKozik16}; the converse
holds assuming $\P\neq\NP$~\cite{dalmau2013robust}. We refer the reader
to~\cite{barto2017polymorphism} for further characterizations.

For Boolean languages, bounded width has the following relational
characterization~\cite{BartoKozik14}.

\begin{proposition}[Boolean bounded-width criterion]
\label{prop:bounded-width-equivalences}
A Boolean constraint language $\Gamma$ has bounded width if and only if $\Gamma$ is $0$-valid, $1$-valid, Horn, dual-Horn, or bijunctive.
\end{proposition}

\subsubsection{Primitive-Positive Definitions}
\label{subsec:pp-definitions}

A primitive-positive definition expresses a relation by an existentially
quantified conjunction of constraints. We use such definitions to express parity
relations over the original language $\Gamma$ when proving the lower bound on
average sensitivity for approximation algorithms.

\begin{definition}[Primitive-positive definition]
Let $S\subseteq\bits^k$. A \emph{primitive-positive definition} (pp definition) of
$S$ over $\Gamma$ is a formula of the following form~\cite{jeavons1997closure}:
\[
  S(x_1,\ldots,x_k)
  \quad\Longleftrightarrow\quad
  \exists z_1,\ldots,z_t\;
  \bigwedge_{j=1}^m \alpha_j,
\]
where $t,m\geq0$. Each $\alpha_j$ is an \emph{atomic formula}, or
\emph{atom}, either of the form $R(y_1,\ldots,y_{\arity(R)})$, where
$R\in\Gamma$, or of the form $y=y'$. All variables are among the displayed free
and existential variables, and repetitions are allowed. We say a pp
definition is \emph{equality-free} if it contains no equality atoms.
\end{definition}

A Boolean relation $R\subseteq\bits^k$ is \emph{self-dual} if, for every $\mathbf{a}\in\bits^k$, we have $\mathbf{a}\in R$ if and only if $\neg\mathbf{a}\in R$, where $\neg\mathbf{a}$ denotes the coordinatewise bit complement of $\mathbf{a}$. The following proposition follows from Post's classification~\cite{post1941twovalues}; the proof of Lemma~8 in~\cite{atserias2019operator} gives the constant-free expressibility argument, in particular for even-arity parity relations. We use only self-dual affine relations, so no constants are needed.

\begin{proposition}[Affine expressibility outside bounded width]
\label{prop:affine-expressibility-outside-bounded-width}
Let $\Gamma$ be a finite Boolean constraint language that does not have
bounded width. Then every nonempty self-dual affine Boolean relation
has a pp definition over $\Gamma$.
\end{proposition}

A $k$-ary relation $R$ is \emph{irredundant} if, for every pair of distinct coordinates $i,j$, there exists $\mathbf a\in R$ such that $a_i\ne a_j$. Equivalently, $R$ does not force any two distinct coordinates to be equal in every tuple. The following equality-elimination lemma is standard~\cite{dalmau2013robust}.
\begin{lemma}[Equality-free pp definitions for irredundant relations]\label{lem:equality-free}
    If an irredundant relation $R$ is pp-definable over a finite constraint language $\Gamma$, then $R$ has an equality-free pp definition over $\Gamma$.
\end{lemma}

\subsection{Average Sensitivity}
\label{subsec:average-sensitivity}

We study the stability of Boolean CSPs using the notion of \emph{average sensitivity} from~\cite{VarmaYoshida23}. We first give some basic definitions.

For a finite set $V$ and $x,y\in\bits^V$, their (unnormalized) \emph{Hamming
distance} is the number of differing positions:
\[
  \distH(x,y)\defeq |\{v\in V:x_v\ne y_v\}|.
\]

A \emph{coupling} of two distributions $\mu$ and $\nu$ over the same space $\Omega$ is a joint distribution over $\Omega\times \Omega$ whose first marginal is $\mu$ and whose second marginal is $\nu$.

We use $W_1(\cdot,\cdot)$ to denote the \emph{$1$-Wasserstein distance} (also known as Earth mover's distance) with respect to the Hamming metric. Thus, for probability distributions $\mu,\nu$ on $\bits^V$, define
\[
  \Wass(\mu,\nu)
  \defeq \min_{\pi\in\Pi(\mu,\nu)}
       \=E_{(X,Y)\sim\pi}[\distH(X,Y)],
\]
where $\Pi(\mu,\nu)$ is the set of couplings of $\mu$ and $\nu$.

We now give the formal notion of \emph{average sensitivity} adapted to Boolean CSPs. Unless stated otherwise, a randomized algorithm $\+A$ takes as input a Boolean CSP instance $\+I=(V,\+C)$ and outputs an assignment $x\in \{0,1\}^V$.

\begin{definition}[Average sensitivity]
For a Boolean CSP instance $\+I=(V,\+C)$ with $\+C\ne\emptyset$ and a
(randomized) algorithm $\+A$, define
\begin{equation}
  \ASens(\+A,\+I)
  \defeq \frac1{|\+C|}\sum_{c\in\+C}
      \Wass\negthinspace\left(
        \Law(\+A(\+I)),\Law(\+A(\+I-c))
      \right)
  \label{eq:constraint-asens}
\end{equation}
as the average sensitivity of $\+A$ on input $\+I$. We set $\ASens(\+A,\+I)\defeq0$ when $\+C=\emptyset$.

\end{definition}

\section{Dichotomy for Stable Solvability}
\label{sec:stable-solvability}

In this section, we prove the dichotomy for stable solvability of Boolean CSPs, beginning with its formal statement.

\begin{theorem}[Dichotomy for stable solvability of Boolean CSPs]
\label{thm:stable-solvability-dichotomy}
For every finite Boolean constraint language $\Gamma$, exactly one of the following holds:
\begin{enumerate}[label=(\roman*),leftmargin=2.2em]
  \item $\Gamma$ has finite duality, and there is a deterministic algorithm
        $\+A_\Gamma$ that solves $\CSP(\Gamma)$ and satisfies
        $\ASens(\+A_\Gamma,\+I)=O_\Gamma(1)$ for every
        satisfiable instance $\+I$ of $\CSP(\Gamma)$.
  \item $\Gamma$ does not have finite duality, and there exists an unbounded set $\+N_\Gamma\subseteq\=N$ such that, for every randomized
        algorithm $\+A$ that solves $\CSP(\Gamma)$ and every $n\in\+N_\Gamma$, there exists a satisfiable instance $\+I$ of $\CSP(\Gamma)$ such that
        $|V(\+I)|=n$ and
        $\ASens(\+A,\+I)=\Omega_\Gamma(n)$.
\end{enumerate}
\end{theorem}

\subsection{Proof of the Upper Bound}
We first prove the upper-bound direction of \Cref{thm:stable-solvability-dichotomy}, which follows naturally from the characterization of finite duality in \Cref{prop:boolean-finite-duality-characterization}.

\begin{proof}[Proof of \Cref{thm:stable-solvability-dichotomy}\textup{(i)}]
Suppose that $\Gamma$ has finite duality. By \Cref{prop:boolean-finite-duality-characterization}, it suffices to consider the following cases:
\begin{itemize}
    \item If $\Gamma$ is $0$-valid (resp., $1$-valid), the constant all-zero (resp., all-one) assignment satisfies every instance. This trivial algorithm clearly has average sensitivity $0$.
    
    \item If $\Gamma$ is strictly essentially positive, fix a CNF definition of each relation using unit clauses and positive clauses. On a satisfiable instance, replace each constraint by the clauses of its fixed CNF definition. Assign $0$ to variables appearing in negative unit clauses of the resulting formula and $1$ to all remaining variables. Every negative unit clause is satisfied. Any other clause must contain a variable not forced to $0$, since otherwise the instance would be unsatisfiable; hence it too is satisfied. Thus the variables forced to $0$ are exactly those witnessed by negative unit clauses in this expansion. Removing a constraint can remove such a witness only for variables in its scope, so the average sensitivity is at most $r_\Gamma$. For strictly essentially negative languages, interchange $0$ and $1$ and use positive unit clauses as witnesses.
\end{itemize}
This concludes the proof of \Cref{thm:stable-solvability-dichotomy}\textup{(i)}.
\end{proof}

\subsection{Proof of the Lower Bound}
We first relate the distance between two deleted constraints in a minimally unsatisfiable instance to the Hamming distance between the resulting solutions. We then construct such instances with linearly many constraints and large diameter. Comparing outputs through the common instance obtained by deleting both constraints turns this separation into an average-sensitivity lower bound.

\subsubsection{Minimally Unsatisfiable Instances}

The next lemma shows that solutions of two different single-constraint deletions must disagree along a path joining those constraints. This will let us convert incidence-graph distance into separation of the algorithm's output distributions.

\begin{definition}[Minimally unsatisfiable instances]\label{def:minimal-unsatisfiable}
We say that a Boolean CSP instance $\+I=(V,\+C)$ is
\emph{minimally unsatisfiable} if $\+I$ is unsatisfiable but
$\+I-c$ is satisfiable for every constraint
$c\in\+C$.
\end{definition}

\begin{definition}[Incidence graph]\label{def:incidence-graph}
    For a Boolean CSP instance $\+I$, its associated \emph{incidence graph} $G_{\mathrm{inc}}(\+I)$ is the bipartite graph whose nodes are the variables and
constraints of $\+I$, with an edge between $c$ and $v$ if
$v\in\var{c}$. Write $d_{G_{\mathrm{inc}}(\+I)}$ for its graph
distance.
\end{definition}

The following lemma shows that the disagreement set contains the variable nodes of a path between the deleted constraints.

\begin{lemma}[Disagreement along a path]
\label{lem:disagreement-propagation}
Let $\+I=(V,\+C)$ be a minimally unsatisfiable Boolean CSP instance in
$\CSP(\Gamma)$, and let $c,c'\in\+C$ be distinct constraints.  Then, for all
$x\in\Sol(\+I-c)$ and
$y\in\Sol(\+I-c')$, the incidence graph
$G_{\mathrm{inc}}(\+I)$ contains a path from $c$ to $c'$ all of whose
variable nodes are variables on which $x$ and $y$ differ.
\end{lemma}

\begin{proof}
Fix $x\in\Sol(\+I-c)$ and $y\in\Sol(\+I-c')$, and define the disagreement set $D\defeq\{v\in V:x_v\ne y_v\}$.

We first claim that $D\cap\var{c}\ne\emptyset$. Otherwise,
$x(\mathbf{s}_c)=y(\mathbf{s}_c)$; since $y$ satisfies $c$, so does $x$,
contradicting the unsatisfiability of $\+I$. Similarly, we have
$D\cap\var{c'}\ne\emptyset$.

Let $H$ be the subgraph of $G_{\mathrm{inc}}(\+I)$ obtained by retaining every constraint
node and only the variable nodes in $D$.  We claim that $c$ and $c'$ are
connected in $H$.  Otherwise, let $K$ be the component of $H$ containing
$c$, and define an assignment $z$ by taking the values of $y$ on the
variable nodes of $K$ and the values of $x$ everywhere else.  Every
constraint in $K$ sees its scoped tuple from $y$, while every constraint
outside $K$ sees its scoped tuple from $x$.  In particular, $z$ satisfies
$c$ using $y$ and satisfies $c'$ using $x$; all other constraints are
satisfied by both assignments.  Thus $z$ satisfies $\+I$, contradicting its
unsatisfiability.  Hence a path from $c$ to $c'$ in $H$ has all of its
variable nodes in $D$.
\end{proof}

The lemma immediately yields the following corollary.
\begin{corollary}
\label{cor:critical-repair-distance}
Let $\+I=(V,\+C)$ be a minimally unsatisfiable Boolean CSP instance in
$\CSP(\Gamma)$, and let $c,c'\in\+C$ be distinct constraints.  Then, for all
$x\in\Sol(\+I-c)$ and
$y\in\Sol(\+I-c')$,
\[
 \distH(x,y)
 \ge \frac12\,d_{G_{\mathrm{inc}}(\+I)}(c,c').
\]
\end{corollary}

\subsubsection{Construction of the Hard Instance}

For every integer $N\geq1$, we need a minimally unsatisfiable instance with $O_\Gamma(N)$ variables and constraints and a connected incidence graph of diameter at least $2N$. The $O_\Gamma(N)$ upper bound on the number of constraints ensures that averaging over deletions retains the separation supplied by a long path. We obtain such an instance by adapting the ``products of links and squares'' construction of Larose et al.~\cite{larose2007characterization}, used there to characterize finite duality through bounded diameters of critical obstructions.

\begin{definition}[Products of links and squares]\label{def:hard-instance-duality}
    For any Boolean constraint language $\Gamma$ and any integer $N\geq 1$, we define the Boolean CSP instance $\+I_N=(V_N,\+C_N)$ as follows:
    \begin{itemize}
        \item \textbf{Variables:} For every $i\in\{0,\ldots,N\}$ and $a,b\in\bits$, we introduce a variable $x_{i,a,b}$. These variables are distinct except for the following identifications at indices $0$ and $N$:
        \[
          x_{0,a,0}=x_{0,a,1}\quad(a\in\bits),
          \qquad
          x_{N,0,b}=x_{N,1,b}\quad(b\in\bits).
        \]
        For $0\leq i\leq N$, write
        \[
          L_i\defeq\{x_{i,a,b}:a,b\in\bits\}
        \]
        for the $i$th layer.  Thus
        $V_N=\bigcup_{i=0}^N L_i$.

        \item \textbf{Constraints:} For every relation $R\in\Gamma$ of arity $r$, every $\mathbf{a}=(a_1,\ldots,a_r),\mathbf{b}=(b_1,\ldots,b_r)\in R$, every $j\in\{1,\ldots,N\}$, and every $\boldsymbol{\lambda}=(\lambda_1,\ldots,\lambda_r)\in\{j-1,j\}^r$, we add the following constraint to $\+C_N$:
        \[
          \bigl(
            (x_{\lambda_1,a_1,b_1},\ldots,
             x_{\lambda_r,a_r,b_r}),
            R
          \bigr).
        \]
    \end{itemize}
\end{definition}

Part (i) of the following lemma is a Boolean specialization of~\cite[Theorems 2.5 and 4.7]{larose2007characterization}, and part (ii) follows from the projection argument in~\cite[Lemma 4.6]{larose2007characterization}. For completeness, we give a direct proof of (i) using the Boolean characterization of finite duality in \Cref{prop:boolean-finite-duality-characterization}, and include the projection argument for (ii).
\begin{lemma}[Properties of the instance $\+I_N$]
\label{lemma:hard-instance-correctness}
Suppose that $\Gamma$ does not have finite duality.  For every $N\geq1$, the
instance $\+I_N=(V_N,\+C_N)$ in \Cref{def:hard-instance-duality} satisfies:
\begin{enumerate}[label=(\roman*),leftmargin=2.2em]
  \item $\+I_N$ is unsatisfiable;
  \item for every $k\in\{0,\ldots,N\}$, the induced subinstance
  $\+I_N[V_N\setminus L_k]$ is satisfiable.
\end{enumerate}
\end{lemma}

\begin{proof}
We first prove (i). Suppose for contradiction that $h\in\Sol(\+I_N)$. For each $i\in\{0,\ldots,N\}$, define the Boolean function $f_i:\bits^2\to\bits$ such that
\[
  f_i(a,b)\defeq h(x_{i,a,b}) \quad (a,b\in \bits).
\]
The equalities $x_{0,a,0}=x_{0,a,1}$ and $x_{N,0,b}=x_{N,1,b}$ imply that $f_0(a,0)=f_0(a,1)$ and $f_N(0,b)=f_N(1,b)$, respectively. Hence, there exist functions $\alpha,\beta:\bits\to\bits$ such that for all $a,b\in \bits$,
\begin{equation}\label{eq:f-initialize}
  f_0(a,b)=\alpha(a),
  \qquad
  f_N(a,b)=\beta(b).
\end{equation}
By the definition of the constraints in \Cref{def:hard-instance-duality}, for every relation $R\in\Gamma$ of arity $r$, every $\mathbf{a}=(a_1,\ldots,a_r),\mathbf{b}=(b_1,\ldots,b_r)\in R$, every $j\in\{1,\ldots,N\}$, and every $\boldsymbol{\lambda}\in\{j-1,j\}^r$, we have
\begin{equation}
  \bigl(f_{\lambda_1}(a_1,b_1),\ldots,
        f_{\lambda_r}(a_r,b_r)\bigr)\in R.
  \label{eq:hard-instance-mixing}
\end{equation}
Taking $j=1$ with every $\lambda_i=0$ and $j=N$ with every $\lambda_i=N$, we find that for every $r$-ary $R\in\Gamma$ and every $\mathbf{a},\mathbf{b}\in R$, both
\[
  (\alpha(a_1),\ldots,\alpha(a_r))
  \quad\text{and}\quad
  (\beta(b_1),\ldots,\beta(b_r))
\]
belong to $R$ (i.e., both $\alpha$ and $\beta$ are endomorphisms of $\Gamma$). Note that if either $\alpha$ or $\beta$ is constant with value $c\in\{0,1\}$, then $\Gamma$ is $c$-valid; by \Cref{prop:boolean-finite-duality-characterization}, this implies $\Gamma$ has finite duality, which contradicts our hypothesis. Thus, $\alpha$ and $\beta$ are both permutations of $\bits$. Since the only permutations on $\bits$ are involutions (the identity and the bitwise complement), they satisfy $\alpha(\alpha(a)) = a$ and $\beta(\beta(b)) = b$.

For each $i\in\{0,\ldots,N\}$, define the transformed Boolean function $g_i:\bits^2\to\bits$ as
\[
  g_i(a,b)\defeq
  f_i\bigl(\alpha(a),\beta(b)\bigr) \quad (a,b\in \bits).
\]
By \eqref{eq:hard-instance-mixing}, since $\alpha(\mathbf{a}), \beta(\mathbf{b}) \in R$, for every $R\in\Gamma$ of arity $r$, every $\mathbf{a},\mathbf{b}\in R$, every $j\in\{1,\ldots,N\}$, and every $\boldsymbol{\lambda}\in\{j-1,j\}^r$, we still have
\begin{equation}\label{eq:hard-instance-mixing-transformed}
      \bigl(g_{\lambda_1}(a_1,b_1),\ldots,
        g_{\lambda_r}(a_r,b_r)\bigr)\in R.
\end{equation}
Moreover, applying the involution property to \eqref{eq:f-initialize}, we have, for all $a,b\in\bits$,
\begin{equation}\label{eq:transformed-initialization}
  g_0(a,b)=a,
  \qquad
  g_N(a,b)=b.
\end{equation}
Hence, $g_0\ne g_N$.

Let $j\in [N]$ be the minimum index such that $g_j\neq g_0$. Fix an $r$-ary $R\in\Gamma$, $\mathbf{a}=(a_1,\ldots,a_r),\mathbf{b}=(b_1,\ldots,b_r)\in R$, and $S\subseteq [r]$. In \eqref{eq:hard-instance-mixing-transformed}, set $\lambda_i=j$ for $i\in S$ and $\lambda_i=j-1$ for $i\notin S$. Because $g_{j-1} = g_0$, the resulting tuple $\mathbf{c}=(c_1,\ldots,c_r)\in R$ is given by
\begin{equation}
  c_i\defeq
  \begin{cases}
    g_j(a_i,b_i),&i\in S,\\
    a_i,&i\notin S
  \end{cases}
  \qquad(1\leq i\leq r).
  \label{eq:first-layer-transfer}
\end{equation}
We then consider the following cases:
\begin{itemize}[leftmargin=2.2em]
  \item If $g_j(0,0)=1$, take $\mathbf{b}=\mathbf{a}$ and $S=\{i:a_i=0\}$ in \eqref{eq:first-layer-transfer}. This shows that $(1,\ldots,1)\in R$ for every $R\in \Gamma$, and therefore $\Gamma$ is $1$-valid.

  \item If $g_j(1,1)=0$, take $\mathbf{b}=\mathbf{a}$ and $S=\{i:a_i=1\}$ in \eqref{eq:first-layer-transfer}. This shows that $(0,\ldots,0)\in R$ for every $R\in \Gamma$, and therefore $\Gamma$ is $0$-valid.

  \item Suppose $g_j(0,0)=0$, $g_j(1,1)=1$, and $g_j(0,1)=1$. Fix an $r$-ary $R\in\Gamma$ and choose $\mathbf{c}\in R$ such that the number of ones in $\mathbf{c}$ is maximized. We claim that every $\mathbf{d}\in R$ satisfies $\mathbf{d}\leq\mathbf{c}$ coordinatewise. Suppose for contradiction that $c_i=0$ while $d_i=1$ for some $1\leq i\leq r$. Applying \eqref{eq:first-layer-transfer} with $\mathbf{a}=\mathbf{c}$, $\mathbf{b}=\mathbf{d}$, and $S=\{i\}$ yields a tuple in $R$ with strictly more ones than $\mathbf{c}$, contradicting its maximality. Hence, the claim holds.

        For any $\mathbf a\in R$ and $\mathbf a\leq\mathbf a'\leq\mathbf c$, applying \eqref{eq:first-layer-transfer} with $\mathbf b=\mathbf c$ and $S=\{i:a_i=0<a_i'\}$ yields $\mathbf a'\in R$.

        Therefore, after fixing the coordinates where $c_i=0$ to zero, $R$ is upward closed. A positive CNF for the remaining coordinates, together with the negative unit clauses $\neg x_i$ for the fixed coordinates, shows that $R$ is strictly essentially positive.

  \item Suppose $g_j(0,0)=0$, $g_j(1,1)=1$, and $g_j(1,0)=0$. As in the previous case, we can show that for each $r$-ary $R\in\Gamma$, choosing $\mathbf{c}\in R$ such that the number of ones in $\mathbf{c}$ is minimized guarantees that every $\mathbf{d}\in R$ satisfies $\mathbf{d}\geq\mathbf{c}$ coordinatewise.

        For any $\mathbf a\in R$ and $\mathbf c\leq\mathbf a'\leq\mathbf a$, applying \eqref{eq:first-layer-transfer} with $\mathbf b=\mathbf c$ and $S=\{i:a_i=1>a_i'\}$ yields $\mathbf a'\in R$.

        Therefore, after fixing the coordinates where $c_i=1$ to one, $R$ is downward closed. A negative CNF for the remaining coordinates, together with the positive unit clauses $x_i$ for the fixed coordinates, shows that $R$ is strictly essentially negative.
\end{itemize}
By \eqref{eq:transformed-initialization} and $g_j\neq g_0$, it cannot happen that $g_j(0,0)=0$, $g_j(1,1)=1$, $g_j(0,1)=0$, and $g_j(1,0)=1$. Hence, the four cases above are exhaustive. In each case, $\Gamma$ has finite duality by \Cref{prop:boolean-finite-duality-characterization}, a contradiction. Thus, $\+I_N$ is unsatisfiable.

For (ii), fix $k\in\{0,\ldots,N\}$ and assign
\[
  x_{i,a,b}\longmapsto
  \begin{cases}
    a,&i<k,\\
    b,&i>k
  \end{cases}
\]
on $V_N\setminus L_k$. Since every constraint (as in \Cref{def:hard-instance-duality}) lies within the union of two consecutive layers, each remaining constraint in $\+I_N[V_N\setminus L_k]$ lies entirely on one side of $L_k$. It is therefore mapped either to the tuple $\mathbf{a}\in R$ or to the tuple $\mathbf{b}\in R$, and is thus satisfied.
\end{proof}

\begin{lemma}[Minimally unsatisfiable instance with large diameter]
\label{lem:hard-instance-minimal-subinstance}
Suppose that $\Gamma$ does not have finite duality.  There is a constant
$\kappa_\Gamma$ such that, for every $N\geq1$, there is a minimally
unsatisfiable instance $\+F_N$ of $\CSP(\Gamma)$ whose incidence graph is connected and which satisfies
\[
 |V(\+F_N)|\leq4N,
 \qquad
 |\+C(\+F_N)|\leq\kappa_\Gamma N,
 \qquad
 \operatorname{diam}(G_{\mathrm{inc}}(\+F_N))\geq2N.
\]
\end{lemma}

\begin{proof}
By \Cref{lemma:hard-instance-correctness}\textup{(i)}, the instance $\+I_N$ is unsatisfiable. Choose an inclusion-minimal subset $\+E \subseteq \+C_N$ for which the constraint subinstance $\+I_N[\+E]$ remains unsatisfiable, and define
\[
  \+F_N\defeq\left(\bigcup_{c\in\+E}\var{c},\+E\right).
\]
By definition, $\+F_N$ is minimally unsatisfiable.

We first prove the diameter bound for $\+F_N$. The incidence graph $G_{\mathrm{inc}}(\+F_N)$ must be connected; otherwise, the constraints in each connected component would form a proper subset of $\+E$ and, by minimality, would admit a satisfying assignment. Combining these assignments would satisfy $\+F_N$, a contradiction.

Recall the layers $L_0,L_1,\dots,L_N$ of $\+I_N$ in \Cref{def:hard-instance-duality}. Note that we must have $V(\+F_N) \cap L_k \neq \emptyset$ for each $0\leq k\leq N$ by \Cref{lemma:hard-instance-correctness}\textup{(ii)} and the unsatisfiability of $\+F_N$. Choose $u \in V(\+F_N) \cap L_0$ and $v \in V(\+F_N) \cap L_N$. By construction in \Cref{def:hard-instance-duality}, a constraint generated for index $j$ contains variables only from $L_{j-1} \cup L_j$. Consequently, any path from $u$ to $v$ must traverse at least $N$ constraint nodes in $G_{\mathrm{inc}}(\+F_N)$, implying $\operatorname{diam}(G_{\mathrm{inc}}(\+F_N)) \geq 2N$.

Finally, the variable count is bounded by $|V(\+F_N)|\leq|V_N|=4N$. Set
\[
    \kappa_\Gamma\defeq\sum_{R\in\Gamma}2^{\arity(R)}|R|^2.
\]
The construction in \Cref{def:hard-instance-duality} then gives
$|\+C(\+F_N)|\leq|\+C_N|=\kappa_\Gamma N$.
\end{proof}

We now prove the lower-bound direction of \Cref{thm:stable-solvability-dichotomy}, thereby completing the dichotomy for stable solvability.

\begin{proof}[Proof of \Cref{thm:stable-solvability-dichotomy}\textup{(ii)}]
Suppose that $\Gamma$ does not have finite duality. Fix an algorithm $\+A$ that solves $\CSP(\Gamma)$ and an integer $n\geq 8$. Set $N\defeq\lfloor n/4\rfloor$ and let $\+F_N$ be the minimally unsatisfiable instance given by \Cref{lem:hard-instance-minimal-subinstance}. By the diameter bound in \Cref{lem:hard-instance-minimal-subinstance}, we can choose a shortest path $P$ in the incidence graph $G_{\mathrm{inc}}(\+F_N)$ whose length is at least $2N$ and let $c_1,\ldots,c_L$ be its constraint nodes, listed in their sequence of appearance on $P$.

Because $P$ alternates between variable and constraint nodes and has length at least $2N$, it must contain at least $N$ constraint nodes; thus, $L\geq N$. Furthermore, because any subpath of a shortest path is itself a shortest path, the distance between constraint nodes in the incidence graph satisfies $d_{G_{\mathrm{inc}}(\+F_N)}(c_i,c_j)=2|i-j|$.

Let $W$ be a set of $n-|V(\+F_N)|$ new variables.  For $1\leq i\leq L$,
define the instance
\[
  \+I_i\defeq
  \bigl(V(\+F_N)\cup W,\,\+C(\+F_N)\setminus\{c_i\}\bigr).
\]
For distinct $i,j\in\{1,\ldots,L\}$, define
\[
  D_{ij}\defeq
  \Wass\bigl(\Law(\+A(\+I_i)),\Law(\+A(\+I_i-c_j))\bigr).
\]
Because $\+F_N$ is minimally unsatisfiable, each subinstance $\+I_i$ is
satisfiable. Since $\+A$ correctly outputs a solution for every $\+I_i$,
the restriction of its output to $V(\+F_N)$ solves $\+F_N-c_i$.
Consequently, \Cref{cor:critical-repair-distance} combined with
$d_{G_{\mathrm{inc}}(\+F_N)}(c_i,c_j)=2|i-j|$ implies that
\[
  \Wass\bigl(\Law(\+A(\+I_i)),\Law(\+A(\+I_j))\bigr)\geq|i-j|.
\]
Moreover, since removing $c_j$ from $\+I_i$ yields exactly the same instance
as removing $c_i$ from $\+I_j$ (i.e., $\+I_i-c_j=\+I_j-c_i$), the resulting
distributions are identical. The triangle inequality therefore yields
\[
  D_{ij}+D_{ji}
  \geq\Wass\bigl(\Law(\+A(\+I_i)),\Law(\+A(\+I_j))\bigr)
  \geq|i-j|.
\]
In the definition of average sensitivity \eqref{eq:constraint-asens} for the instance $\+I_i$, the term corresponding to the removal of constraint $c_j$ is exactly $D_{ij}$. Since all terms in the sum are nonnegative and $|\+C(\+I_i)|\leq\kappa_\Gamma N$, we can bound the sum of the average sensitivities over the path constraints:
\[
  \sum_{i=1}^L\ASens(\+A,\+I_i)
  \geq\frac{1}{\kappa_\Gamma N}
        \sum_{1\leq i<j\leq L}(D_{ij}+D_{ji})
  \geq\frac{1}{\kappa_\Gamma N}
        \sum_{1\leq i<j\leq L}(j-i)
  =\frac{L^3-L}{6\kappa_\Gamma N}.
\]
Dividing by $L$ to average over the $L$ instances and applying the bounds $L\geq N$, $N\geq 2$, and $N\geq n/8$, we deduce that there exists $i\in\{1,\ldots,L\}$ such that
\[
  \ASens(\+A,\+I_i)
  \geq\frac{L^2-1}{6\kappa_\Gamma N}
  \geq\frac{N^2-1}{6\kappa_\Gamma N}
  \geq\frac{n}{96\kappa_\Gamma}.
\]
Finally, $|V(\+I_i)|=n$ by construction. This proves the
$\Omega_{\Gamma}(n)$ bound for every $n\geq8$.
\end{proof}

\section{Dichotomy for Stable Approximability}
\label{sec:stable-approximability}

In this section, we prove the dichotomy for stable approximability of Boolean CSPs, beginning with its formal statement.

\begin{theorem}[Dichotomy for stable approximability of Boolean CSPs]
\label{thm:stable-approximability-dichotomy}
For every finite Boolean constraint language $\Gamma$,
exactly one of the following holds:
\begin{enumerate}[label=(\roman*),leftmargin=2.2em]
  \item $\Gamma$ has bounded width, and for every $\eps\in(0,1]$, there is
        a randomized algorithm $\+A$ that
        $(1-\eps)$-approximates $\CSP(\Gamma)$ and, for every satisfiable
        instance $\+I=(V,\+C)$ of $\CSP(\Gamma)$ with $|V|=n\geq 2$,
        \[
          \ASens(\+A,\+I)=O_\Gamma(\eps^{-1}\log n).
        \]
  \item $\Gamma$ does not have bounded width, and there exist
        $\eps_\Gamma\in(0,1]$ and an unbounded set
        $\+N_\Gamma\subseteq\=N$ such that, for every randomized algorithm
        $\+A$ that $(1-\eps_\Gamma)$-approximates $\CSP(\Gamma)$ and every
        $n\in\+N_\Gamma$,
        there exists a satisfiable instance $\+I=(V,\+C)$ of $\CSP(\Gamma)$ such that
        $|V|=n$ and
        $\ASens(\+A,\+I)=\Omega_\Gamma(n)$.
\end{enumerate}
\end{theorem}

We first prove the upper-bound direction of \Cref{thm:stable-approximability-dichotomy}, using the characterization of bounded width in \Cref{prop:bounded-width-equivalences}.

\subsection{Proof of the Upper Bound: Horn and Dual-Horn Languages}
\label{sec:horn}

Let $\Gamma$ be a finite Horn (resp., dual-Horn) constraint language. It is well known that every satisfiable instance $\+I$ of $\CSP(\Gamma)$ admits a unique coordinatewise least (resp., greatest) solution, denoted by $x^{\min}(\+I)$ (resp., $x^{\max}(\+I)$)~\cite{horn1951sentences}, and that such extremal solutions can be computed in linear time~\cite{dowling1984linear}.

The least solution of a Horn-SAT instance can be highly unstable under constraint deletion. To illustrate this, consider the chain instance $\+I$ on $n$ variables $x_0,x_1,\dots,x_{n-1}$, consisting of one unit constraint $x_0=1$ and $n-1$ implication constraints $x_i\implies x_{i+1}$ for $0\le i<n-1$. The algorithm that outputs $x^{\min}(\+I)$ has average sensitivity $\Omega(n)$ on this instance.

If approximate solutions are permitted, a natural strategy to stabilize the output is to introduce random perturbations: retain each constraint independently with a fixed probability $p\in[1-\eps,1]$ and output the coordinatewise least solution of the resulting subinstance.

Choosing one value of $p$ does not by itself bound the average sensitivity. Instead, we draw $p\sim\Unif([1-\eps,1])$. As $p$ increases, we can couple the retained subinstances so that constraints are only added and the least solution is coordinatewise nondecreasing. Each variable changes at most once, which bounds the total change over the interval by the number of non-isolated variables. Averaging over $p$ then gives the desired bound on average sensitivity. The formal algorithm proceeds as follows:

\begin{definition}[Algorithm for Horn and dual-Horn languages]\label{def:alg-horn}
    Fix $\eps\in(0,1]$, and let $\Gamma$ be a finite Horn (resp., dual-Horn) Boolean constraint language. The following algorithm takes as input a satisfiable instance $\+I=(V,\+C)$ of $\CSP(\Gamma)$ and outputs an assignment:
    \begin{enumerate}[label=(\arabic*),leftmargin=2.2em]
    \item Draw the probability $p$ of retaining each constraint uniformly from $[1-\eps,1]$. \label{item:alg-horn-1}
    \item Include every constraint independently with probability
    $p$, obtaining $\+S\subseteq\+C$. \label{item:alg-horn-2}
    \item Output $x^{\min}(\+I[\+S])$ (resp., $x^{\max}(\+I[\+S])$).
\end{enumerate}
\end{definition}

We next establish the approximation and stability guarantees for this algorithm.

\begin{theorem}[Horn and dual-Horn upper bounds]
\label{thm:horn-dual-horn}
Let $\Gamma$ be a finite Horn or dual-Horn constraint language.  For every
$\eps\in(0,1]$, the algorithm $\+A$ in \Cref{def:alg-horn}
$(1-\eps)$-approximates $\CSP(\Gamma)$ and has average sensitivity $O_{\Gamma}(\eps^{-1})$ for every satisfiable instance $\+I$ of $\CSP(\Gamma)$.
\end{theorem}

\begin{proof}
Fix the instance $\+I=(V,\+C)$. We prove only the Horn case; the dual-Horn case is analogous.

We first prove the approximation guarantee. Every subinstance of the satisfiable Horn instance $\+I$ is satisfiable and Horn, so the output satisfies every retained constraint. Hence, each constraint $c\in\+C$ is violated with probability at most $\eps$. Summing over all constraints proves the approximation guarantee.

We now bound the average sensitivity. Note that the average sensitivity is trivially $0$ when $\+C = \emptyset$ by definition; therefore, we may assume $\+C \ne \emptyset$. The heart of the proof will be the following monotonicity property: for Horn languages, the least solutions are coordinatewise nondecreasing under insertion of constraints.

For each subset of constraints $\+E\subseteq\+C$, let
\[
  h(\+E)\defeq\sum_{v\in V}x^{\min}(\+I[\+E])_v
\]
denote the number of ones in the least solution of $\+I[\+E]$. Note that $h$ is nondecreasing.

We use the following realization of step~\ref{item:alg-horn-2} of
\Cref{def:alg-horn}, with explicitly specified randomness:
\begin{itemize}
    \item independently for each constraint $c\in\+C$, draw $p_c$ uniformly at random from $[0,1]$ and include $c$ if and only if $p_c\leq p$.
\end{itemize}

For each $x\in[0,1]$, let
\[
\+S_x\defeq \{c\in \+C:p_c\leq x\}, \qquad F(x)\defeq \E{h(\+S_x)},
\]
where the expectation is over the priorities $(p_c)_{c\in\+C}$. For $x\in[1-\eps,1]$, $\+S_x$ has the distribution of the retained set conditional on $p=x$, and $F(x)$ is the expected number of ones in the resulting least solution. For each $c\in\+C$, let $\Delta_c(x)$ denote the expected Hamming distance between the least solutions when $c$ is forced to be present and when it is forced to be absent; formally,
\[
\Delta_c(x)\defeq \E{h(\+S_x\cup \{c\})-h(\+S_x\setminus \{c\})}.
\]
To differentiate $F$, give each constraint occurrence $c$ its own retention probability $x_c\in[0,1]$. The expectation of $h$ is a real multilinear polynomial in $(x_c)_{c\in\+C}$, and its partial derivative with respect to $x_c$ is the expected difference between forcing $c$ to be present and forcing it to be absent. Setting all $x_c=x$ and applying the chain rule gives
\[
F'(x)=\sum\limits_{c\in \+C}\Delta_c(x).
\]

We couple the executions of the algorithm on the original instance and on the instance obtained by deleting a constraint $c\in\+C$ by using the same $p$ and the same priorities $(p_d)_{d\in\+C\setminus\{c\}}$. For $x\in[1-\eps,1]$, conditional on $p=x$, the outputs before and after deleting $c$ agree when $p_c>x$; otherwise, their expected Hamming distance is
$\Delta_c(x)$. Hence
\[
  \Wass\bigl(\Law(\+A(\+I)),\Law(\+A(\+I-c))\bigr)
  \leq\frac1\eps\int_{1-\eps}^1x\Delta_c(x)\,dx.
\]
Let $U\defeq\bigcup_{c\in\+C}\var{c}$ denote the set of non-isolated variables of $\+I$. Summing over $c$ and noting that monotonicity gives $F(1)-F(0)\leq|U|$, we obtain
\[
  \ASens(\+A,\+I)
  \leq\frac{1}{\eps |\+C|}\int_{1-\eps}^1xF'(x)\,dx
  \leq\frac{|U|}{\eps |\+C|}
  \leq\frac{r_\Gamma}{\eps}.
\]
Here, the last inequality uses $|U|\leq r_\Gamma |\+C|$. This concludes the proof.

\end{proof}
\subsection{Proof of the Upper Bound: Bijunctive Languages}
\label{sec:bijunctive}

We next give our algorithm for stably approximating Boolean CSPs over bijunctive languages.

\subsubsection{Warm-up: $2$-Coloring}
\label{par:two-coloring-warmup}
As a warm-up to our algorithmic construction for bijunctive languages, we first present an algorithm for stably approximating $2$-coloring. The $2$-coloring problem is a Boolean CSP over the bijunctive language consisting only of the disequality relation. It is also the canonical example of a problem that is not stably solvable, thus motivating the relaxation to stable approximability~\cite{VarmaYoshida23}.

A satisfiable $2$-coloring instance $\+I=(V,\+C)$ is represented by a bipartite multigraph $G=(V,E)$, with one edge occurrence for each disequality constraint. Parallel edges preserve constraint multiplicities, and deleting one constraint deletes only the corresponding edge occurrence. We identify $\+I$ with this multigraph throughout the argument.

Our algorithm for stably approximating $2$-coloring uses the \emph{low-diameter decomposition} with exponential shifts introduced by Miller, Peng, and Xu~\cite{miller2013parallel}. This partitions the graph into clusters, each with a root. We properly color each cluster by assigning color $0$ to its root and coloring each vertex by the parity of its distance from that root.

An edge can be violated only if its endpoints belong to different clusters, so bounding the probability of this event gives the approximation guarantee. To bound average sensitivity, we use the same exponential shifts before and after deleting an edge. A vertex can change color only if the deleted edge lies on a chosen shortest path from its root to that vertex in the original graph. Thus, for each vertex, the number of edge deletions that can change its color is at most the length of this path. We bound the expected path lengths to obtain the bound on average sensitivity.

\begin{definition}[Algorithm for $2$-coloring]
\label{def:alg-two-coloring}
Fix $\beta=\eps\in(0,1]$. Given a bipartite graph $G=(V,E)$, the algorithm performs the following steps. Write $d$ for the shortest-path distance in $G$, with $d(u,v)=+\infty$ when $v$ is unreachable from $u$.
\begin{enumerate}[label=(\arabic*),leftmargin=2.2em]
    \item Independently sample a shift $\delta_u\sim\ExpDist(\beta)$ of rate $\beta$ for each vertex $u\in V$.
    \item For each vertex $v\in V$, define its \emph{root} by
    \[
        r(v)=\arg\max_{u\in V}\{\delta_u-d(u,v)\},
    \]
    resolving ties via a fixed ordering of $V$. We call $\delta_u-d(u,v)$ the \emph{score} of $u$ at $v$.
    \item Output the color assignment $x_v = d(r(v),v)\bmod 2$ for every $v\in V$.
\end{enumerate}
\end{definition}

The following bound, given by the proof of~\cite[Lemma~4.4]{miller2013parallel}, will be crucial for our analysis of both $2$-coloring and $2$-SAT.
\begin{lemma}[Gap between the two largest shifted exponential values]
\label{lem:exponential-gap}
Let $q\ge2$, let $d_1,\ldots,d_q\in\mathbb R$ be fixed, and let
$\delta_1,\ldots,\delta_q$ be independent exponential random variables
of rate $\beta>0$.  Write $Z_j=\delta_j-d_j$ and let
$Z_{(1)}\ge Z_{(2)}$ be the two largest values among $Z_1,\ldots,Z_q$.
Then, for every $t\ge0$,
\[
    \Pr{Z_{(1)}-Z_{(2)}\le t}
    \le 1-\e^{-\beta t}\le\beta t.
\]
\end{lemma}

We now show that the algorithm in \Cref{def:alg-two-coloring} has the claimed approximation and stability guarantees for $2$-coloring.

\begin{proposition}[Stable approximation of $2$-coloring]
\label{prop:two-coloring}
For every $\eps\in(0,1]$ and bipartite graph $G=(V,E)$ with
$n=|V|\ge2$, the algorithm $\+A$ in
\Cref{def:alg-two-coloring} $(1-\eps)$-approximates the $2$-coloring problem on $G$ and has average sensitivity $O(\eps^{-1}\log n)$.
\end{proposition}
\begin{proof}[Proof of \Cref{prop:two-coloring}]
Fix $\eps\in(0,1]$ and a bipartite graph $G=(V,E)$ with $n\ge2$,
and set $\beta=\eps$ as in \Cref{def:alg-two-coloring}.

We first prove the approximation guarantee. An edge can be
violated only if it is inter-cluster, that is, if its endpoints have different roots. For an edge $\{v,w\}$, let $K\subseteq V$ be the vertex set of its connected component. It is inter-cluster only if the two largest values among $\{\delta_u-\min\{d(u,v),d(u,w)\}:u\in K\}$ differ by at most one. Here $|K|\geq2$ and all distances in the displayed set are finite. Therefore, by \Cref{lem:exponential-gap}, the violation
probability is at most $\beta=\eps$.
Summing over edges gives the approximation guarantee of at least $1-\eps$.

We then prove the bound on average sensitivity. We couple the executions on $G$ and $G-e$ by using the same exponential shifts $\delta_u$ for all $u\in V$.
For each $v\in V$, fix a shortest path $P_v$ from
$r(v)$ to $v$ in $G$.  If $e\notin P_v$, the score of $r(v)$ at $v$
is unchanged, while all other root scores can only decrease.
Therefore, $v$ retains its root and color.  Since the maximization defining $r(v)$
includes $u=v$,
\[
    \delta_{r(v)}-|P_v|\ge\delta_v>0,
    \qquad |P_v|\le\max_{u\in V}\delta_u-\delta_v.
\]
Since $\E{\max_{u\in V}\delta_u}=O(\beta^{-1}\log n)$, we have
$\E{|P_v|}=O(\beta^{-1}\log n)$.
Each vertex $v$ can change its color under at most $|P_v|$ edge
deletions in this coupling.  Consequently, for $|E|>0$,
\[
    \ASens(\+A,G)
    \le\frac{1}{|E|}\sum_{v\in V}\E{|P_v|}
    =O(\eps^{-1}\log n),
\]
where the last estimate uses the facts that isolated vertices have
$|P_v|=0$ and that a graph with $|E|>0$ has at most $2|E|$ non-isolated
vertices. This proves the proposition.
\end{proof}

\subsubsection{From $2$-Coloring to $2$-SAT}

We now present our stable approximation algorithm for general bijunctive languages. The core approach naturally extends the $2$-coloring algorithm by operating on the implication graph associated with $2$-SAT instances.

To use a literal as a root for propagating truth along directed implication paths, we require it to be feasible, meaning that some satisfying assignment sets it to true. In a satisfiable $2$-SAT instance, this is equivalent to there being no path from the literal to its negation. Deleting a constraint can make new literals feasible, so we must control both path changes and changes to the candidate roots. We address the latter by using nested random subinstances: one for computing directed distances and a larger one for testing feasibility.

Fix a finite bijunctive constraint language $\Gamma$. For each relation $R\in\Gamma$, fix a defining $2$-CNF formula over its arguments, and let $b_\Gamma\ge1$ be an upper bound, depending only on $\Gamma$, on the number of clauses in any such formula. After substituting the scope variables, discard tautologies and simplify a clause $(a\vee a)$ to the unit clause $(a)$. Given a Boolean CSP instance $\+I=(V,\+C)$ over $\Gamma$, replace each constraint $c=R(v_1,\ldots,v_k)$ with this simplified $2$-CNF formula, treating a unit clause $(a)$ as the logically equivalent disjunction $(a\vee a)$ when constructing the implication graph. Deleting a constraint removes all arcs contributed by its expansion. We then construct the standard implication graph for the resulting $2$-CNF formula~\cite{aspvall1979linear}.

\begin{definition}[Implication graph]
\label{def:implication-graph}
The \emph{implication graph} $G_{\mathrm{imp}}(\+I)$ of a Boolean CSP instance $\+I=(V,\+C)$ (with a fixed 2-CNF formula representation) is the directed multigraph with vertex set
\[
    V(G_{\mathrm{imp}}(\+I)) = \{x, \neg x : x \in V\}
\]
and, for each clause $(a\vee b)$ in the expansion of $\+C$, the directed arcs
\[
    \neg a \to b \quad \text{and} \quad \neg b \to a.
\]
\end{definition}

Let $\mathsf{Feas}(\+I)$ denote the set of literals that evaluate to true in at least one satisfying assignment of $\+I$. Our algorithm for bijunctive languages is detailed in \Cref{def:alg-bijunctive}. To stabilize the additional sensitivity introduced by newly feasible literals, the algorithm employs two nested random subinstances $\+I^{(1)}\subseteq\+I^{(2)}$: a fixed threshold on constraint priorities determines the implication arcs in $\+I^{(1)}$, whereas a larger randomized cutoff defines the subinstance $\+I^{(2)}$ used only to test root feasibility. 

\begin{definition}[Algorithm for bijunctive languages]
\label{def:alg-bijunctive}
Fix $\eps\in(0,1]$ and let
$\beta=\eps/(2b_\Gamma)$.
Given a satisfiable Boolean CSP instance $\+I=(V,\+C)$ over $\Gamma$,
the algorithm performs the following steps:
\begin{enumerate}[label=(\arabic*),leftmargin=2.2em]
\item Draw independent priorities $\tau_c\sim\Unif([0,1])$ for all
$c\in\+C$ and an independent cutoff $T\sim\Unif([1-\eps/2,1])$.
Form the two subinstances
\[
    \+I^{(1)}=\+I[\{c\in\+C:\tau_c\le 1-\eps/2\}],\qquad
    \+I^{(2)}=\+I[\{c\in\+C:\tau_c\le T\}].
\]
\item Independently draw shifts $\delta_a\sim\ExpDist(\beta)$ for all
literals $a$ on $V$.
\item Write $d(u,v)$ for the directed shortest-path
distance in $G_{\mathrm{imp}}(\+I^{(1)})$, with unit-length arcs and
$d(u,v)=+\infty$ when $v$ is unreachable from $u$. For each literal $\ell$, define its \emph{score} by
\[
    \operatorname{score}(\ell)
    =\max_{a\in\mathsf{Feas}(\+I^{(2)})}
        \bigl(\delta_a-d(a,\ell)\bigr),
\]
and output the assignment setting each $x_i\in V$ to $1$ if
$\operatorname{score}(x_i)>\operatorname{score}(\neg x_i)$, and to $0$ otherwise.
\end{enumerate}
\end{definition}

We now show that the algorithm in \Cref{def:alg-bijunctive} has the claimed approximation and stability guarantees for bijunctive languages.
\begin{theorem}[Bijunctive upper bound]
\label{thm:bijunctive}
For every finite bijunctive Boolean constraint language $\Gamma$ and
$\eps\in(0,1]$, the algorithm $\+A$ in \Cref{def:alg-bijunctive}
$(1-\eps)$-approximates $\CSP(\Gamma)$ and has average sensitivity $O_\Gamma(\eps^{-1}\log n)$
for every satisfiable Boolean CSP instance $\+I=(V,\+C)$ over $\Gamma$ with $n=|V|\ge2$.
\end{theorem}

\begin{proof}
Fix a satisfiable instance $\+I=(V,\+C)$.
A literal in $\mathsf{Feas}(\+I^{(2)})$ cannot reach both $x_i$ and
$\neg x_i$ in $G_{\mathrm{imp}}(\+I^{(1)})$: a solution of
$\+I^{(2)}$ setting that literal to true would then set both $x_i$
and $\neg x_i$ to true.
For each variable $x_i$, at least one of its literals is feasible
and has distance zero to itself.  Hence
\[
  \max\{\operatorname{score}(x_i),\operatorname{score}(\neg x_i)\}>0
  \quad\text{almost surely}.
\]
By independence and continuity of the shifts, this maximum is
attained by a unique feasible starting literal almost surely.
We call this starting literal the \emph{root} of $x_i$.

We first prove the approximation guarantee.
Condition on the subinstances $\+I^{(1)}$ and $\+I^{(2)}$. Every unit
clause $(a)$ in the $2$-CNF expansion of $\+I^{(1)}$ is satisfied by the output. Indeed, $a$ is
feasible in $\+I^{(2)}$ and has positive score, whereas no feasible
literal can reach $\neg a$: the arc $\neg a\to a$ would then make it
reach both $a$ and $\neg a$.

For each non-unit clause $(a\vee b)$ in the expansion of
$\+I^{(1)}$, the
arcs $\neg a\to b$ and $\neg b\to a$ give
\[
    \operatorname{score}(b)\ge \operatorname{score}(\neg a)-1,
    \qquad
    \operatorname{score}(a)\ge \operatorname{score}(\neg b)-1.
\]
If the output violates $(a\vee b)$, then
$\operatorname{score}(\neg a)$ and $\operatorname{score}(\neg b)$
are positive and differ by at most one.
The feasible literals reaching $\neg a$ and those reaching $\neg b$
form disjoint sets: a solution of $\+I^{(2)}$ setting a literal
in both sets to true would violate $(a\vee b)$.
If either set is empty, the clause cannot be violated. Otherwise, consider the values $\delta_r-d(r,\neg a)$ for feasible $r$ reaching
$\neg a$ and $\delta_r-d(r,\neg b)$ for feasible $r$ reaching
$\neg b$.
These values are independent, with maxima
$\operatorname{score}(\neg a)$ and $\operatorname{score}(\neg b)$
over the respective sets.
Thus a violation requires the two largest of these values to differ
by at most one, which has probability at most $1-\e^{-\beta}$
by \Cref{lem:exponential-gap}.
Each original constraint is omitted from $\+I^{(1)}$ with probability
$\eps/2$ and expands into at most $b_\Gamma$ clauses.  Its violation
probability is therefore at most
$\eps/2+b_\Gamma(1-\e^{-\beta})\leq \eps$, proving the approximation guarantee.

We then prove the bound on average sensitivity.
The case $\+C=\varnothing$ is immediate.  Otherwise, for each
$c\in\+C$, couple the executions on $\+I$ and $\+I-c$ using the
same priorities $(\tau_d)_{d\in\+C\setminus\{c\}}$ for the remaining constraints, the same cutoff $T$, and the same shifts $(\delta_a)$.
Compare the outputs in two steps: first replace the implication graph
by $G_{\mathrm{imp}}(\+I^{(1)}-c)$, keeping
$\mathsf{Feas}(\+I^{(2)})$ fixed; then replace the set of feasible roots
by $\mathsf{Feas}(\+I^{(2)}-c)$, keeping
$G_{\mathrm{imp}}(\+I^{(1)}-c)$ fixed.
Deleting a constraint absent from a subinstance has no effect.
The triangle inequality bounds the Hamming distance between the
coupled outputs by the sum of the Hamming distances in these two steps.
Let $U\defeq\bigcup_{c\in\+C}\var{c}$ be the set of non-isolated
variables; the output on variables outside $U$ is unchanged.

For the first step, fix $x_i\in U$ and a shortest path $P_i$ in
$G_{\mathrm{imp}}(\+I^{(1)})$ from its root $r$ to the literal
$\ell\in\{x_i,\neg x_i\}$ set to true by the output.  Then
\[
    \delta_r-|P_i|=\operatorname{score}(\ell)>0,
    \qquad |P_i|<\delta_r\le\max_a\delta_a.
\]
If deleting $c$ removes no arc of $P_i$, the contribution
$\delta_r-d(r,\ell)$ is unchanged, while all competing contributions
can only decrease.  Thus the assigned value of $x_i$ is unchanged.
At most $|P_i|$ constraint deletions can remove an arc of $P_i$.
As in the $2$-coloring proof, the maximum of the $2n$ shifts has
expectation $O(\beta^{-1}\log n)$, so
$\E{|P_i|}=O(\beta^{-1}\log n)$.
Summing over $x_i\in U$, the expected sum of Hamming distances
in the first step over all deletions is $O(|U|\beta^{-1}\log n)$.

For the second step, first fix $\+I^{(1)}=(V,\+C^{(1)})$.
Choose a uniformly random ordering $c_1,\ldots,c_k$ of the constraints in
$\+C\setminus\+C^{(1)}$, and put
$\+I_j^{(2)}=\+I[\+C^{(1)}\cup\{c_1,\ldots,c_j\}]$.
For $1\le j\le k$, let $\Delta_j$ be the Hamming distance between
the assignments obtained using $\mathsf{Feas}(\+I_{j-1}^{(2)})$
and $\mathsf{Feas}(\+I_j^{(2)})$ as the sets of feasible roots,
with the graph $G_{\mathrm{imp}}(\+I^{(1)})$ and shifts shared.  We claim
\begin{equation}\label{eq:bijunctive-prefix-changes}
    \E{\sum_{j=1}^k\Delta_j}=O(|U|\log n).
\end{equation}
Conditional on $\+I^{(1)}$, the expectation in \eqref{eq:bijunctive-prefix-changes} is over the random ordering and the shifts. To prove the claim, also condition on the ordering and fix $x_i\in U$.
For every $a\in\mathsf{Feas}(\+I_0^{(2)})$ that reaches $x_i$ or
$\neg x_i$, let $d_a$ be the distance to the literal it reaches; it
cannot reach both.  Only candidates with $\delta_a-d_a>0$ can be the
root of $x_i$, since a literal made true by a solution of the original
instance remains feasible in every prefix and contributes its positive
shift at distance zero.  Conditional on the set of candidates with
$\delta_a-d_a>0$, the memoryless property implies that their positive
values $\delta_a-d_a$ are independent rate-$\beta$ exponentials.
The sets $\mathsf{Feas}(\+I_j^{(2)})$ decrease as $j$ increases. Fix, independently of the shifts, an arbitrary ordering of the candidate roots removed at each step. The root can change only when
the largest remaining value is removed. Among $s$ i.i.d.
continuous values, each is largest with probability $1/s$; hence the
expected number of root changes is at most
$H_{2n}=\sum_{s=1}^{2n}1/s=O(\log n)$.  The value of $x_i$ changes no
more often than its root.  Summing over $x_i\in U$ proves
\eqref{eq:bijunctive-prefix-changes}.

We now relate the comparisons between assignments for consecutive subinstances to constraint deletions.
Conditional on $\+C^{(1)}$, the remaining priorities and $T$ are
independent uniform samples on $[1-\eps/2,1]$.  Their relative order is
uniform, and the rank of $T$ is uniform on $\{0,\ldots,k\}$; thus
$\+I^{(2)}$ is a uniform choice among the $k+1$ prefixes.
For deletions outside $\+C^{(1)}$, only the $j$ retained constraints can change the assignment when deleted from $\+I_j^{(2)}$. By symmetry, choosing one of them uniformly produces the same distribution as the pair of prefixes of lengths $j-1$ and $j$. Hence, the expected sum of Hamming distances caused by these deletions is $j\E{\Delta_j}$.
Averaging over the prefix length, the expected sum is
$\frac{1}{k+1}\sum_{j=1}^k j\E{\Delta_j}
\le\E{\sum_{j=1}^k\Delta_j}=O(|U|\log n)$.

For deletions inside $\+C^{(1)}$, we average over $\+C^{(1)}$ and group
the comparisons by the
remaining set $\+E=\+C^{(1)}\setminus\{c\}$.
The graph used in the second step is now $G_{\mathrm{imp}}(\+I[\+E])$.
Since each constraint belongs to $\+C^{(1)}$ independently with
probability $1-\eps/2$,
\[
    \Pr{\+C^{(1)}=\+E\cup\{c\}}
    =\frac{2-\eps}{\eps}\Pr{\+C^{(1)}=\+E}
    \qquad(c\notin\+E).
\]
Conditional on $\+C^{(1)}=\+E\cup\{c\}$, regard $T$ as a
priority assigned to $c$.  Together with the priorities of
$\+C\setminus(\+E\cup\{c\})$, it defines a uniform ordering
of $\+C\setminus\+E$.
The two sets of feasible roots being compared are those of the prefixes
immediately before and after $c$ in this ordering.
For fixed $\+E$, summing over $c\notin\+E$ therefore counts each
comparison between the assignments for two consecutive subinstances once in expectation. Thus,
\eqref{eq:bijunctive-prefix-changes}, applied with $\+I[\+E]$ as the fixed first subinstance, bounds this sum.
Multiplying by the probability ratio above and averaging over
$\+E$ bounds the expected sum over deletions inside $\+C^{(1)}$ by
$O(|U|\eps^{-1}\log n)$.

Combining the bounds for the two steps, applying the coupling bound
for $W_1$, and dividing by $|\+C|$ gives
\[
    \ASens(\+A,\+I)
    =O\left(\frac{|U|}{|\+C|}(\beta^{-1}+\eps^{-1})\log n\right)
    =O_\Gamma(\eps^{-1}\log n),
\]
where we use $|U|\le r_\Gamma|\+C|$ and $\beta=\eps/(2b_\Gamma)$.

\end{proof}

By \Cref{prop:bounded-width-equivalences}, a bounded-width Boolean
language is $0$-valid, $1$-valid, Horn, dual-Horn, or bijunctive.
The constant assignments resolve the first two cases with sensitivity
zero, while \Cref{thm:horn-dual-horn} and \Cref{thm:bijunctive} resolve the others. This gives the proof of \Cref{thm:stable-approximability-dichotomy}\textup{(i)}.

\subsection{Proof of the Lower Bound}
\label{sec:stable-approximability-lower-bound}
We now prove the lower-bound direction of the dichotomy for stable approximability (\Cref{thm:stable-approximability-dichotomy}), thereby completing its proof.

To transfer a lower bound from a constraint language $\Lambda$ to $\Gamma$, we replace each constraint over $\Lambda$ by constraints over $\Gamma$ using an equality-free pp definition. The difficulty is that one deletion before the reduction removes several constraints afterward, while average sensitivity controls only an average of single deletions. Randomly deleting the replacement constraints lets us compare the intermediate instances with conditioned samples of one distribution. We combine this reduction with a family of instances over four-variable parity relations. We obtain this family by adapting the LTC construction of Yoshida and Zhang~\cite{YoshidaZhang26}. Earlier reductions addressed worst-case sensitivity via the PCP framework~\cite{FlemingYoshida26}.

\subsubsection{Transferring Average-Sensitivity Lower Bounds}

For the instances used in the lower bound, the reduction reserves the same collection of disjoint sets of new existential variables before and after one constraint deletion, and assigns distinct sets to distinct constraints uniformly at random. The reduced instances therefore have the same variable set. The reduction also deletes every replacement constraint independently. Independence lets us bound the sum of Wasserstein distances between output distributions along the sequence that removes the replacement constraints of one original constraint. Averaging this bound over the deleted original constraint bounds the average sensitivity of the algorithm on original instances in terms of the expected average sensitivity of the algorithm on reduced instances.

\begin{lemma}[Transfer of average-sensitivity lower bounds]
\label{lem:pp-average-transfer}
Let $a\geq1$ and $b\geq0$ be integers, and let $\Lambda$ and $\Gamma$ be finite Boolean constraint languages such that
\begin{enumerate}[label=(\roman*),leftmargin=2.2em]
    \item
    Every relation of $\Lambda$ has an equality-free pp definition over $\Gamma$ with at most $a$ atomic formulas and at most $b$ existential variables.
    \item There are constants $\eps_0\in(0,1)$ and $c_0>0$, an integer
    $D\geq1$, and an unbounded set $\+N\subseteq\mathbb N_{\geq1}$
    such that, for every randomized algorithm $\+B$ that
    $(1-\eps_0)$-approximates $\CSP(\Lambda)$ and every $n\in\+N$,
    there is a satisfiable Boolean CSP instance $\+I=(V,\+C)$ over
    $\Lambda$ satisfying
    \[
        |V|=n,\qquad |\+C|\leq Dn,\qquad
        \ASens(\+B,\+I)\geq c_0n.
    \]
\end{enumerate}
Set $\eta\defeq\eps_0/(2a)$.
For every randomized algorithm $\+A$ that $(1-\eta)$-approximates
$\CSP(\Gamma)$ and every $n\in\+N$, there is a
satisfiable Boolean CSP instance $\+I'=(V',\+C')$ over $\Gamma$ satisfying
\[
    |V'|=(1+bD)n,\qquad |\+C'|\leq aDn,
\]
and
\[
    \ASens(\+A,\+I')\geq\frac{c_0\eta^{a-1}}{a}\,n.
\]
\end{lemma}

\begin{proof}
Fix an equality-free pp definition over $\Gamma$ for every relation in $\Lambda$. Let $\+A$ be a fixed randomized algorithm that $(1-\eta)$-approximates $\CSP(\Gamma)$.

Given a Boolean CSP instance $\+{I}=(V,\+{C})$ over $\Lambda$, we construct a random instance $\+{T}(\+{I})$ over $\Gamma$ as follows:
\begin{enumerate}[label=(\arabic*),leftmargin=2.2em]
    \item \textbf{Variables:} The variable set $V'$ of $\+{T}(\+{I})$ is defined as follows:
    \begin{enumerate}[label=(\alph*)]
        \item The original variable set $V$;
        \item $\max\{D|V|, |\+{C}|\}$ disjoint sets of $b$ new variables. We denote the variables in the $i$th set by $z_{i,1},\dots,z_{i,b}$.
    \end{enumerate}
    
    \item \textbf{Constraints:} The constraint multiset $\+{C}'$ of $\+{T}(\+{I})$ is generated via the following steps:
    \begin{enumerate}[label=(\alph*),ref=2(\alph*)]
        \item Choose the indices $(i_c)_{c\in\+C}$ uniformly among all choices in which distinct constraints receive distinct indices.\label{item:pp-group-assignment}
        \item For each constraint $c=((v_1,\dots,v_k),R)\in\+C$, substitute $v_1,\dots,v_k$ for the free variables in the fixed pp definition of $R$, and substitute $z_{i_c,j}$ for its $j$th existential variable. Add every resulting atomic formula as a constraint to an intermediate multiset $\+E$.\label{item:pp-replacement-constraints}
        \item Construct the final multiset $\+{C}'$ by independently deleting each constraint in $\+{E}$ with probability $\eta$.\label{item:pp-random-deletions}
    \end{enumerate}
\end{enumerate}

Define $\+{B}$ as the algorithm that, given an input instance $\+{I}$, samples the random instance $\+{T}(\+{I})$, executes $\+{A}$ on $\+{T}(\+{I})$, and returns its output assignment restricted to the original variables $V$.

We first verify that $\+B$ converts the $(1-\eta)$-approximation guarantee of $\+A$ into a $(1-2a\eta)$-approximation guarantee. Suppose $\+I$ is satisfiable. Any satisfying assignment for $\+{I}$ naturally extends to a satisfying assignment for the intermediate multiset $\+{E}$ by assigning the new variables values that witness the corresponding pp definitions. Because deleting constraints can only relax the instance, this extended assignment also satisfies $\+C'$. Hence, $\+T(\+I)$ is satisfiable.

If $\+B(\+I)$ violates a constraint $c\in\+C$ and none of its replacement constraints were deleted, then $\+A(\+T(\+I))$ must violate at least one of the $\Gamma$-constraints obtained from the pp definition of $c$.
For each constraint, the probability that at least one of its replacement constraints was deleted is at most $a\eta$. Therefore we have
\[
    \E{\viol_{\+I}(\+B(\+I))}
    \leq a\eta|\+C|
        +\E{\viol_{\+T(\+I)}(\+A(\+T(\+I)))}
    \leq 2a\eta|\+C|
    =\eps_0|\+C|,
\]
where the second inequality uses the approximation guarantee of $\+A$ on each
sampled instance and $|\+E|\leq a|\+C|$.
Thus $\+B$ $(1-\eps_0)$-approximates $\CSP(\Lambda)$.

We now bound the average sensitivity of $\+B$ in terms of the expected average sensitivity of $\+A$ on $\+T(\+I)$. Fix an instance $\+{I}=(V,\+{C})$ such that $\+{C}\neq\emptyset$ and $|\+{C}|\leq D|V|$, and fix $c\in\+C$. Because both $\+{I}$ and $\+{I}-c$ contain at most $D|V|$ constraints, their constructions use the same $D|V|$ sets of new variables and therefore have the same variable set $V'$.

We construct a coupling between the generation of $\+{T}(\+{I})$ and $\+{T}(\+{I}-c)$ as follows:
\begin{itemize}
    \item First, we couple the choices of indices in step~\ref{item:pp-group-assignment} so that each constraint $c'\in \+C\setminus\{c\}$ receives the same index $i_{c'}$.
    \item Then, for every $c'\in\+C\setminus\{c\}$, we couple the deletions so that the same subset of the replacement constraints obtained from its pp definition is deleted in step~\ref{item:pp-random-deletions}.
\end{itemize}
Consequently, $\+{T}(\+{I}-c)$ is obtained directly from $\+{T}(\+{I})$ by deleting the retained replacement constraints generated by $c$.

Let $c_1,\dots,c_t$, where $t\leq a$, be the
replacement constraints generated by the pp definition of $c$, before the independent deletions in step~\ref{item:pp-random-deletions}. For
$0\leq i\leq t$, let $\+C_i=\{c_1,\dots,c_i\}$ and
$\+I_i=(V',\+C'\setminus\+C_i)$, with $\+C_0=\emptyset$.
Then $\+I_0=\+T(\+I)$ and $\+I_t=\+T(\+I-c)$.
Since restriction to $V$
cannot increase Hamming distance, the triangle inequality gives
\begin{equation}
    \Wass\tp{\Law(\+B(\+I)),\Law(\+B(\+I-c))}\leq\sum_{i=1}^t
        \Wass\tp{\Law(\+A(\+I_{i-1})),\Law(\+A(\+I_i))}.
\label{eq:pp-deletion-path}
\end{equation}

We next bound each term on the right-hand side for $1\le i\le t$.
Fix the indices chosen in step~\ref{item:pp-group-assignment}; this also fixes $\+E$ and all replacement constraints. We adopt the convention that $\+T(\+I)-e=\+T(\+I)$ whenever $e\notin\+C'$.
Observe that
\[
    \Law(\+I_{i-1}) = \Law(\+T(\+I) \mid \+C'\cap\+C_{i-1}=\emptyset),
\]
and this conditioning event occurs with probability $\eta^{i-1}$.
For each realized pair $(\+I_{i-1},\+I_i=\+I_{i-1}-c_i)$, take an optimal coupling of the corresponding output distributions and then average these couplings. Since the conditioning event has probability $\eta^{i-1}$ and the Wasserstein distance is nonnegative, the conditional expectation is at most $\eta^{-(i-1)}$ times the unconditional expectation. Thus,
\begin{equation}\label{eq:single-wasserstein-bound}
    \Wass\tp{\Law(\+A(\+I_{i-1})), \Law(\+A(\+I_i))} \le \eta^{-(i-1)} \E{\Wass\tp{\Law(\+A(\+T(\+I))), \Law(\+A(\+T(\+I)-c_i))}}.
\end{equation}
In \eqref{eq:single-wasserstein-bound}, the laws on the left average over both the independent constraint deletions and the internal randomness of $\+A$, with the chosen indices fixed. On the right, the expectation is over the deletions, and each law inside it is conditional on the sampled input and averages only over $\+A$'s internal randomness. Averaging this inequality over the random choice of indices and using convexity of Wasserstein distance for mixtures gives the corresponding bound for the unconditional laws in \eqref{eq:pp-deletion-path}. All subsequent expectations average over both the chosen indices and the deletions.

Combining the definition of average sensitivity with \eqref{eq:pp-deletion-path} and \eqref{eq:single-wasserstein-bound}, we obtain
\begin{equation}\label{eq:pp-average-transfer}
\begin{aligned}
    \,\ASens(\+B,\+I)
    &\leq\frac{1}{|\+C|}\eta^{-(a-1)}
        \E{\sum_{e\in\+C'}
            \Wass\tp{\Law(\+A(\+T(\+I))),\Law(\+A(\+T(\+I)-e))}}\\
    &=\frac{1}{|\+C|}\eta^{-(a-1)}
        \=E\left[|\+C'|\,\ASens(\+A,\+T(\+I))\right]\\
    &\leq a\eta^{-(a-1)}
        \=E\left[\ASens(\+A,\+T(\+I))\right].
\end{aligned}
\end{equation}

For every $n\in\+N$, the hypothesis of the lemma supplies a satisfiable
instance $\+I$ for $\+B$ with $n$ variables, at most $Dn$ constraints,
and $\ASens(\+B,\+I)\geq c_0n$.
By \eqref{eq:pp-average-transfer}, some sampled instance $\+T(\+I)$
must satisfy
\[
    \ASens(\+A,\+T(\+I))
    \geq \frac{c_0\eta^{a-1}}{a}\,n.
\]
Note that every sampled $\+T(\+I)$ is satisfiable, has exactly $(1+bD)n$ variables
and at most $aDn$ constraints, thereby proving the stated bounds.
\end{proof}

We next state a lower bound on average sensitivity for algorithms approximating the parity CSP from which we transfer the lower bound. For $b\in\bits$, define four-variable parity relations
\[
    R_b\defeq\{(x_1,x_2,x_3,x_4)\in\bits^4:
        x_1\oplus x_2\oplus x_3\oplus x_4=b\},
    \qquad \Lambda_4\defeq\{R_0,R_1\}.
\]

The following proposition gives a lower bound on the average sensitivity of algorithms that approximate $\CSP(\Lambda_4)$. Its proof uses the LTC construction of Yoshida and Zhang and their lower bound on the Wasserstein distance between output distributions, averaged over pairs of messages that differ in one bit~\cite[Lemmas~3.7 and~4.4]{YoshidaZhang26}. We defer the proof to Section~\ref{sec:parity-source-proof}.

\begin{proposition}[Average-sensitivity lower bound for approximating $\CSP(\Lambda_4)$]
\label{prop:parity-source-lower-bound}
There are constants $\eps_0\in(0,1)$ and $c_0>0$, a constant integer
$D_0\geq1$, and an unbounded set $\+N_0\subseteq\mathbb N$ such that, for every randomized
algorithm $\+A$ that $(1-\eps_0)$-approximates $\CSP(\Lambda_4)$
and every $n\in\+N_0$, there is a satisfiable Boolean CSP instance
$\+I=(V,\+C)$ over $\Lambda_4$ with $|V|=n$ and $|\+C|\leq D_0n$
satisfying
\[
    \ASens(\+A,\+I)\geq c_0n.
\]
The constants $\eps_0,c_0,D_0$ are independent of $n$ and $\+A$.
\end{proposition}

We are now ready to prove \Cref{thm:stable-approximability-dichotomy}\textup{(ii)}.

\begin{proof}[Proof of \Cref{thm:stable-approximability-dichotomy}\textup{(ii)}]
Fix a finite Boolean constraint language $\Gamma$ that lacks bounded width. Both relations in $\Lambda_4$ are nonempty, self-dual, and affine. Consequently, by \Cref{prop:affine-expressibility-outside-bounded-width}, they admit pp definitions over $\Gamma$.

We use the following property of each $R_b$ with $b\in\{0,1\}$:
\begin{equation}\label{eq:not-equal}
\text{for every }i\neq j,\quad
\text{there exists }\mathbf a\in R_b\text{ such that }a_i\neq a_j.
\end{equation}
For $R_1$, take the tuple having $a_i=1$ and all other coordinates $0$.
For $R_0$, choose $k\notin\{i,j\}$ and take the tuple having
$a_i=a_k=1$ and all other coordinates $0$. Thus both $R_0$ and $R_1$
are irredundant, and their pp definitions can be made equality-free by
\Cref{lem:equality-free}.

Fix such definitions of $R_0$ and $R_1$ with at most $a$ atomic formulas and at most $b$
existential variables. Set
\[
  \eps_\Gamma\defeq\frac{\eps_0}{2a},
  \qquad
  \+N_\Gamma\defeq\{(1+bD_0)n:n\in\+N_0\}.
\]
Apply \Cref{lem:pp-average-transfer} to the lower bound for $\CSP(\Lambda_4)$ in
\Cref{prop:parity-source-lower-bound}, taking $\Lambda=\Lambda_4$,  $D=D_0$, and $\+N=\+N_0$. For every $n\in\+N_0$, the resulting
instance has $(1+bD_0)n\in\+N_\Gamma$ variables and sensitivity
$\Omega_\Gamma(n)$. This sensitivity is linear in the number of variables of the resulting instance. This proves
the lower bound.
\end{proof}

\subsubsection{Average-Sensitivity Lower Bound for Parity CSPs}
\label{sec:parity-source-proof}
We prove \Cref{prop:parity-source-lower-bound} by adapting the construction of Yoshida and Zhang~\cite{YoshidaZhang26}. Their locally testable codes induce a satisfiable CSP instance for each message. For any algorithm that violates at most a sufficiently small fraction of constraints in expectation on each instance, the expected Wasserstein distance between its output distributions on two instances whose messages differ in one bit is linear in the block length. The expectation is over a uniformly chosen message and a uniformly chosen bit to flip.

To express the construction using self-dual four-variable parity relations, we represent each code coordinate by the XOR of two variables and split each resulting six-variable test into two four-variable constraints. One additional constraint fixes each message bit. Instances for messages differing in one bit then yield the same subinstance after deleting one constraint from each. Comparing their output distributions through this subinstance converts the bound on Wasserstein distance into an average-sensitivity lower bound.

\begin{definition}[Error-correcting code]
\label{def:error-correcting-codes}
Let $n,k\in \=N$ and $n\geq k\geq1$. A \emph{(binary) error-correcting code}, or simply a \emph{code}, of length $n$ and message length $k$ is a pair $(C,f)$, where $C\subseteq \bits^n$ and the encoder $f:\bits^k\to C$ is a bijection. The \emph{rate} of the code is $k/n$,
and its \emph{(relative) distance} is $\frac1n\min\limits_{c,c'\in C,c\neq c'}\distH(c,c')$. We often identify the code with the subset $C$ when the encoder is clear from context.

A code $(C,f)$ is called \emph{linear} if $C$ is a linear subspace of $\=F^n_2$. The \emph{(absolute) distance} of a linear code equals the minimum Hamming weight of a nonzero codeword.

A code $(C,f)$ of length $n$ and message length $k$ is called \emph{systematic} if $f(x)$ has the same first $k$ bits as $x$ for any $x\in \{0,1\}^k$.
\end{definition}

A locally testable code admits a randomized tester that reads a small number of bits and rejects a string with probability at least proportional to its relative distance from the code. For a binary string $c\in \{0,1\}^V$ and a subset $S\subseteq V$, we denote by $c_S$ the restriction of $c$ to $S$. We also write $\distH(c,C)=\min\limits_{c'\in C}\distH(c,c')$ for a nonempty subset $C\subseteq \{0,1\}^V$. We use the formulation in~\cite{YoshidaZhang26}, which also bounds the number of tests in which each coordinate occurs.

\begin{definition}[Locally testable code]\label{def:locally-testable-codes}
    Let $\kappa\in\=R^{>0}$ and let $q,D\geq1$ be integers. We say an error-correcting code $(C,f)$ of length $n$ is a \emph{locally testable code (LTC)} with detection probability $\kappa$, query complexity $q$, and variable degree (at most) $D$ if there exists a finite nonempty multiset $\+S$ of nonempty subsets $S\subseteq [n]$ with $|S|=q$, where each subset $S$ is associated with a set $V_S\subseteq \=F^S_2$ of \emph{allowed local views}, such that the following properties hold:
    \begin{itemize}
        \item If $c\in C$, then $c_S\in V_S$ for every $S\in \+S$.
        \item For every $c\in \=F^n_2$,
        \[
        \Pr[S\sim\Unif(\+S)]{c_S\notin V_S}\geq \kappa\cdot\frac{\distH(c,C)}{n}.
        \]
        \item Every coordinate $i\in [n]$ belongs to at most $D$ subsets $S\in \+S$.
    \end{itemize}
\end{definition}

We use the bounded-degree, systematic version of the $c^3$-LTC construction~\cite{dinur2022locally} obtained in~\cite[Lemmas~2.5 and~4.4]{YoshidaZhang26}. Here $c^3$ denotes constant rate, constant relative distance, and constant query complexity. The proof of their Lemma~4.4 gives the tests on three coordinates specified below.

\begin{lemma}[Bounded-degree systematic $c^3$-LTC]
\label{lem:yz-bounded-degree-parity}
There exist constants $r,\delta,\kappa>0$, an integer $D\geq1$,
and an infinite family of linear systematic LTCs
$(C_n,f_n)$ with $C_n\subseteq\=F_2^n$ such that each code in the family
has rate at least $r$, relative distance at least $\delta$,
detection probability at least $\kappa$, variable degree at most $D$, and
query complexity $3$, with each local test of the form $x_u\oplus x_v\oplus x_w=0$ on three distinct coordinates $u,v,w$.
\end{lemma}

An LTC naturally induces a CSP instance by regarding each coordinate of a received word as a variable and each test performed by the tester as a constraint on those variables. We give the following definition specialized to systematic LTCs.

\begin{definition}[CSP induced by a systematic LTC]\label{def:csp-induced-LTC}
Let $(C,f)$ be a systematic LTC of block length $n$ and message
length $k$, with local test multiset $\+S$ and allowed local views $(V_S)_{S\in\+S}$. Assume that $S\setminus[k]\ne\emptyset$ for every test occurrence $S\in\+S$. For any message $\sigma\in \{0,1\}^k$, define a CSP instance $\+I_C(\sigma)=(V,\+C)$ as follows:
\begin{itemize}
    \item The variable set of $\+I_C(\sigma)$ is defined as $V\defeq [n]\setminus [k]$.
    \item For each test occurrence $S\in\+S$, write $S\cap V=\{i_1<\cdots<i_t\}$, where $t\geq1$, and add one constraint with scope $(i_1,\ldots,i_t)$ and relation
    \[
    R_{S,\sigma}\defeq\left\{(a_1,\ldots,a_t)\in\bits^t:
    \begin{array}{l}
      \text{there exists }\beta\in V_S\text{ with }
      \beta_{S\cap[k]}=\sigma_{S\cap[k]}\\
      \text{and }\beta(i_j)=a_j\text{ for every }j\in[t]
    \end{array}\right\}.
    \]
    Thus the constraint multiset retains one occurrence per local test.
\end{itemize}
\end{definition}

The next lemma gives a lower bound on the expected Wasserstein distance between the output distributions on a uniformly chosen message and on the message obtained by flipping a uniformly chosen bit. It is proved in~\cite[Lemma~3.7]{YoshidaZhang26}; we express its approximation hypothesis in terms of the induced CSP. The factor $1/n$ converts our unnormalized Hamming distance to the normalization in that source.
\begin{lemma}[Wasserstein distance for neighboring messages]
\label{lem:yz-average-separation}
Let $(C,f)$ be a systematic LTC with block length $n$,
message length $k\geq1$, rate $r$, relative distance $\delta$,
and detection probability $\kappa$, whose tests each query a coordinate outside $[k]$. Fix $\eps\in[0,1]$. Suppose that a randomized algorithm $\+A$, on input $\sigma\in\{0,1\}^k$, produces a $(1-\eps)$-approximation to $\+I_C(\sigma)$. Then
\[
    \frac1n\E[\sigma,i]{
        \Wass\bigl(\Law(\+A(\sigma)),
                   \Law(\+A(\sigma^{(i)}))\bigr)
    }\geq\delta\left(1-\frac{2\eps}{r\kappa}\right)
        -\frac{2\eps}{\kappa}-\frac1n,
\]
where $\sigma$ and $i$ are chosen independently and uniformly from $\bits^k$ and $[k]$, respectively,
and $\sigma^{(i)}$ is $\sigma$ with its $i$-th bit flipped.
\end{lemma}

We are now ready to prove \Cref{prop:parity-source-lower-bound}.

\begin{proof}[Proof of \Cref{prop:parity-source-lower-bound}]
Fix a family $(C_n,f_n)$ and constants $r,\delta,\kappa,D$ from \Cref{lem:yz-bounded-degree-parity}. For a sufficiently large block length $n$ in this family, write $(C,f)=(C_n,f_n)$, let $k=k_n$ be its message length, let $V=[n]$, and let $m=m_n$ be its number of local tests. Every test is a nontrivial three-coordinate parity equation. No test can involve only message coordinates: systematicity lets those coordinates take every possible bit pattern, including one violating that equation. Thus the induced instances of \Cref{def:csp-induced-LTC} are defined and have exactly $m$ positive-arity constraints.

For sufficiently large $n$, we have $\delta n>1$. The lower bound on relative distance thus guarantees that the unit vectors $e_v\notin C$ for all $v\in V$. The detection guarantee therefore implies that some test rejects $e_v$, meaning that $x_v$ must have a nonzero coefficient in at least one test equation. Since every test queries exactly three coordinates and every coordinate is queried by at most $D$ tests, counting occurrences of coordinates in tests yields
\[
    n\leq 3m\leq Dn.
\]

For each message $\sigma\in\bits^k$, we construct an instance $\+I_\sigma=(W,\+C_\sigma)$ over $\Lambda_4$ as follows:
\begin{enumerate}[label=(\arabic*),leftmargin=2.2em]
    \item \textbf{Variables:} The variable set is
    \[
        W\defeq\{v^{(1)},v^{(2)}:v\in V\}
        \cup\{t_c:c\text{ is a local test}\},
    \]
    where $t_c$ is a new variable associated with the local test $c$.
    \item \textbf{Constraints:} The multiset $\+C_\sigma$ consists of the following:
    \begin{enumerate}[label=(\alph*)]
        \item For each local test $c$ with coordinates $u<v<w$, replace each code bit by the XOR of its two representatives, and set
        \[
          (y_{c,1},\ldots,y_{c,6})
          \defeq(u^{(1)},u^{(2)},v^{(1)},v^{(2)},w^{(1)},w^{(2)}).
        \]
        The test becomes
        \[
            y_{c,1}\oplus y_{c,2}\oplus\cdots\oplus y_{c,6}=0.
        \]
        We express this equation by two four-variable parity constraints that share the variable $t_c$:
        \[
            R_0(y_{c,1},y_{c,2},y_{c,3},t_c),\qquad
            R_0(t_c,y_{c,4},y_{c,5},y_{c,6}).
        \]
        \item For every $x\in[k]$, we enforce the condition $x^{(1)}\oplus x^{(2)}=\sigma_x$ by adding the following constraint:
        \[
            c_x^\sigma\defeq
            R_{\sigma_x}(x^{(1)},x^{(2)},x^{(1)},x^{(1)}).
        \]
    \end{enumerate}
\end{enumerate}

For any $\sigma\in\bits^k$, the codeword $f(\sigma)$ satisfies every test. Consequently, setting $v^{(1)}=f(\sigma)_v$ and $v^{(2)}=0$ for all $v\in V$, together with setting $t_c\defeq y_{c,1}\oplus y_{c,2}\oplus y_{c,3}$, satisfies $\+I_\sigma$. Each instance has size
\begin{equation}
    |W|=2n+m=\Theta(n),
    \qquad |\+C_\sigma|=2m+k\leq 5m=O(n).
    \label{eq:parity-source-size}
\end{equation}
Crucially, for every $\sigma\in\bits^k$ and $x\in[k]$, we have
\begin{equation}
    \+I_\sigma-c_x^\sigma
    =\+I_{\sigma^{(x)}}-c_x^{\sigma^{(x)}}.
    \label{eq:parity-source-common-deletion}
\end{equation}

We now establish the average-sensitivity lower bound for this family. Choose $\eps_0\in(0,1)$, depending only on $r,\delta,\kappa,D$, so that $5D\eps_0\leq1$ and
\[
  \delta\left(1-\frac{10D\eps_0}{r\kappa}\right)
       -\frac{10D\eps_0}{\kappa}\geq\frac{\delta}{2}.
\]
Let $\+A$ be a $(1-\eps_0)$-approximation algorithm for $\CSP(\Lambda_4)$. Define an algorithm $\+B$ that approximates the induced instance $\+I_C(\sigma)$ for any $\sigma\in\bits^k$ as follows:
\begin{itemize}
    \item On input $\sigma\in\bits^k$, execute $\+A(\+I_\sigma)$ to obtain an assignment $\tau\in \bits^W$. Let $z_v \defeq \tau_{v^{(1)}}\oplus\tau_{v^{(2)}}$ for each $v\in V$, and output the restriction $z_{V\setminus[k]}$ to the coordinates outside the message.
\end{itemize}

A constraint of $\+I_C(z_{[k]})$ is violated by $z_{V\setminus[k]}$ exactly when the full vector $z$ violates the corresponding local test. Moreover, every message coordinate $x\in[k]$ for which $z_x\ne\sigma_x$ causes $c_x^\sigma$ to be violated. Changing the message from $z_{[k]}$ to $\sigma$ can affect at most $D$ induced constraints per differing coordinate, so
\[
    \viol_{\+I_C(\sigma)}(\+B(\sigma))
    \leq \viol_{\+I_C(z_{[k]})}(z_{V\setminus[k]}) + D \cdot \distH(z_{[k]},\sigma)
    \leq D\viol_{\+I_\sigma}(\tau).
\]
Taking expectations, we obtain
\begin{align*}
    &\E{\viol_{\+I_C(\sigma)}(\+B(\sigma))}\\
   (\text{$\+A$ is a $(1-\eps_0)$-approximation algorithm})\quad \leq &D\eps_0|\+C_\sigma|\\
   (\text{by \eqref{eq:parity-source-size}})\quad \leq &5D\eps_0m.
\end{align*}
Because $\+I_C(\sigma)$ contains precisely $m$ constraints, $\+B$ is a $(1-5D\eps_0)$-approximation algorithm for every induced instance. This satisfies the approximation prerequisite for \Cref{lem:yz-average-separation}.

We now use the lower bound on the average Wasserstein distance between the output distributions of $\+B$ to obtain an average-sensitivity lower bound for $\+A$. Let $\sigma$ and $x$ be independent and uniform on $\bits^k$ and $[k]$, respectively. The decoding map $\tau\mapsto z_{V\setminus[k]}$ cannot increase the $1$-Wasserstein distance. For $n\geq4/\delta$, \Cref{lem:yz-average-separation}, applied with $\eps=5D\eps_0$, yields the following bound. If the code parameters exceed the uniform lower bounds $r,\delta,\kappa$, substituting these lower bounds only weakens the estimate:
\begin{equation}\label{eq:parity-source-separation}
    \=E_{\sigma,x}\left[
        \Wass\bigl(\Law(\+A(\+I_\sigma)),
                   \Law(\+A(\+I_{\sigma^{(x)}}))\bigr)
    \right]
    \geq \=E_{\sigma,x}\left[
        \Wass\bigl(\Law(\+B(\sigma)),
                   \Law(\+B(\sigma^{(x)}))\bigr)
    \right]
    \geq \frac{\delta n}{4}.
\end{equation}

Finally, with the same uniform choices of $\sigma$ and $x$, note that $\sigma$ and $\sigma^{(x)}$ are identically distributed. Applying the triangle inequality through the output distribution on the common instance in~\eqref{eq:parity-source-common-deletion}, obtained by deleting one constraint from each input, gives
\[
    \=E_\sigma[\ASens(\+A,\+I_\sigma)]
    \geq\frac{k}{2|\+C_\sigma|}
        \=E_{\sigma,x}\left[
            \Wass\bigl(\Law(\+A(\+I_\sigma)),
                       \Law(\+A(\+I_{\sigma^{(x)}}))\bigr)
        \right]=\Omega(n)=\Omega(|W|),
\]
where $k\geq rn$ and \eqref{eq:parity-source-size} give constants uniform over the code family. More explicitly, put $M\defeq2D/3+1$ and $L\defeq2+D/3$. Then $|\+C_\sigma|\leq Mn$, $|W|\leq Ln$, and the displayed expectation is at least $r\delta n/(8M)$. Take
\[
  D_0\defeq\lceil M\rceil,\qquad
  c_0\defeq\frac{r\delta}{8ML},\qquad
  \+N_0\defeq\{2n+m_n:n\text{ is a sufficiently large block length in the family}\}.
\]
The set $\+N_0$ is unbounded. For each of its sizes, some message $\sigma$ satisfies $\ASens(\+A,\+I_\sigma)\geq c_0|W|$, and $|\+C_\sigma|\leq D_0|W|$. These constants do not depend on the size or the algorithm, as required.
\end{proof}

\section{Conclusions and Open Problems}

We provide a complete classification, parameterized by the constraint language $\Gamma$, of stable solvability and stable approximability for satisfiable Boolean CSPs, with algorithmic stability quantified by average sensitivity. We establish that $\CSP(\Gamma)$ is stably solvable if and only if $\Gamma$ has finite duality, and stably approximable if and only if $\Gamma$ has bounded width.

Our proofs use the relational characterizations of finite duality and bounded width. For the upper bounds, we give low-average-sensitivity algorithms for Horn, dual-Horn, and bijunctive languages. The lower-bound proof for stable solvability uses sparse minimally unsatisfiable instances with large diameter. For stable approximability, we combine a four-variable parity construction with a reduction that transfers average-sensitivity lower bounds through equality-free pp definitions.

Our work leaves open several avenues for future research regarding the algorithmic stability of constraint satisfaction problems:

\begin{itemize}
    \item Our current average-sensitivity bound for stably approximating bijunctive languages is $O_{\Gamma}(\eps^{-1}\log n)$. Is this logarithmic dependence on $n$ necessary?
    
    \item Another natural next step is to extend the current stability classification to more general settings. Prominent directions include generalizing the dichotomy to CSPs over non-Boolean domains and studying stable approximability of CSPs without the satisfiability promise.

    \item Although we study stability here in terms of average sensitivity, one may ask how the classification would change under worst-case sensitivity or other notions of stability. It would also be interesting to establish formal connections or separations among these and other stability measures.

    \item Our upper bounds currently rely on a case analysis of Boolean constraint languages. Can stable solvability and stable approximability be established directly from finite duality and bounded width, respectively?
\end{itemize}

\section*{AI Disclosure}
During the preparation of this manuscript, the authors utilized ChatGPT 5.6 and 6 as AI assistants. While various algorithmic concepts and proof strategies emerged through extensive interactions with these models, the core theoretical constructions remain strictly human. Specifically, the algorithm and proof for stable approximation of $2$-colorings were developed independently by the authors. Building upon this, the authors designed the generalized algorithm for stably approximating $2$-SAT (bijunctive languages), while the AI models provided assistance in constructing and refining the associated formal proofs.

The authors devoted substantial effort to reorganizing and rewriting the manuscript, including all
key arguments and algorithmic descriptions, to make the ideas and proofs clear and accessible to
researchers. They also verified all mathematical statements, proofs, algorithms, and references, and
take full responsibility for all content.

\ifdoubleblind
\else
\section*{Acknowledgements}
Yuichi Yoshida is supported by JSPS KAKENHI Grant Number 24K02903.
\fi

\bibliographystyle{alpha}
\bibliography{references}

\end{document}